\documentclass[journal,11pt,onecolumn,draftcls]{IEEEtran}%
\usepackage{amsfonts}
\usepackage{amssymb}
\usepackage{amsmath}
\usepackage{graphicx}%
\providecommand{\U}[1]{\protect\rule{.1in}{.1in}}
\newtheorem{theorem}{Theorem}

\newtheorem{lemma}{Lemma}

\newtheorem{proposition}{Proposition}
\newtheorem{remark}{Remark}

\begin{document}
%
\title
{Performance analysis and optimal selection\\of large mean-variance portfolios under estimation risk}
\author{\IEEEauthorblockN{Francisco Rubio*\footnote{*Corresponding author}%
, Xavier Mestre and Daniel P. Palomar}
}
\maketitle
%

\begin{abstract}%
We study the consistency of sample mean-variance portfolios of arbitrarily
high dimension that are based on Bayesian or shrinkage estimation of the input
parameters as well as weighted sampling. In an asymptotic setting where the
number of assets remains comparable in magnitude to the sample size, we
provide a characterization of the estimation risk by providing deterministic
equivalents of the portfolio out-of-sample performance in terms of the
underlying investment scenario. The previous estimates represent a means of
quantifying the amount of risk underestimation and return overestimation of
improved portfolio constructions beyond standard ones. Well-known for the
latter, if not corrected, these deviations lead to inaccurate and overly
optimistic Sharpe-based investment decisions. Our results are based on recent
contributions in the field of random matrix theory. Along with the asymptotic
analysis, the analytical framework allows us to find bias corrections
improving on the achieved out-of-sample performance of typical portfolio
constructions. Some numerical simulations validate our theoretical findings.%
\end{abstract}%

\vspace{1cm}%
\vspace{3cm}%

\section{Introduction%
\label{secIntro}%
}

\subsection{Background and research motivations}

The foundations of modern portfolio theory were laid by Markowitz's
ground-breaking article \cite{M52}, where the idea of diversifying a portfolio
by spreading bets across a universe of risky financial assets was refined and
generalized by the more sophisticated one of combining the assets so as to
optimize the risk-return tradeoff. In practice, Markowitz's mean-variance
optimization framework for solving the canonical wealth allocation problem
relies on the statistical estimation of the unknown expected values and
covariance matrix of the asset returns from sample market observations.

In general, the uncertainty inherently associated with imperfect moments
estimates represents a major drawback in the application of the classical
Markowitz framework. Indeed, the optimal mean-variance solution has been
empirically observed to be significantly sensitive to deviations from the true
input parameters. In addition, and aside from computational complexity issues,
the estimation of the parameters is involved, mainly due to the instability of
the parameter estimates through time. Generally, estimates of the covariance
matrix are more stable than those of the mean returns, and so many studies
disregard the estimation of the latter and concentrate on improving the sample
performance of the so-called global minimum variance portfolio (GMVP); see
arguments in \cite{JM03}.

In the financial literature, the previous source of portfolio performance
degradation is referred to as estimation risk. Especially when the number of
securities is comparable to the number of observations, estimation errors may
in fact prevent the mean-variance optimization framework from being of any
practical use. In fact, for severe levels of estimation risk, the naive
portfolio allocation rule namely obtained by equally weighting the assets
without incorporating any knowledge about their mean and covariance turns out
to represent a firm candidate choice \cite{DGU09}. The consistency and
distributional properties of sample optimal mean-variance portfolios and their
Sharpe ratio performance has been analyzed and characterized for finite
samples and asymptotically (see, most recently, \cite{OS06,KS08,SS10sharpe},
and also the list of references therein for earlier contributions).

Commencing with particularly high activity and contribution levels in the
80's, there exists a vast literature on portfolio selection methods accounting
for estimation risk by explicitly dealing with the lack of robustness and
stability of the sample optimal mean-variance solution, which we do not intend
to exhaustively review here; we refer the reader to \cite{M05B,FKP07B} for a
thorough treatment of the subject. Some remarks on the two main lines of
approach are in order. One class of methods based on convex analysis and
nonlinear optimization techniques focuses on formulations of the allocation
problem where robustness to estimation errors is achieved by means of the
explicit modeling of parameter uncertainty regions. Conceptually, assuming
worst-case bounds on the input parameters may not be effective in practice
since no information is available about the distribution of the estimated
parameters within the uncertainty boundaries.

On the other hand, instances of a second family of methods of statistical or
probabilistic nature are approaches based on Bayesian and Steinian shrinkage
estimation seeking efficiency by weighting a sensible prior belief and the
classical sample estimator in inverse proportion to their dispersion (see,
e.g., \cite{Jo82,LW03}). As a matter of fact, this class of techniques
provides a rather general framework for understanding different forms of
portfolio corrections and performance improvements tackling estimation risk.
Indeed, explicit links have been found between the latter and the robust
optimization solutions introduced above, which turn out to be possibly
interpreted from a shrinkage estimation perspective \cite{S07robust}.
Furthermore, constraining%
\footnote
{Usual constraints on the allocation weights that are typically considered in the portfolio construction process are those modelling the self-financing characteristic of the investment rule, as well as budget and short-selling restrictions.}
the portfolio weights has been additionally noted to be equivalent to adding
some structure to the covariance estimation problem as obtained through
Bayesian or shrinkage-based procedures \cite{JM03}. Arguably, the latter
constitutes an effective way helping to avoid overfitting the sample data and
to improve the stability of the realized portfolio solution out of sample and
over time (see also \cite{DGNU09}, where the authors investigate the effects
of norm-constrains in the solution of the weight vector). The application of
Bayesian and shrinkage approaches is not limited to the moment estimation
problem alone, but can indeed be extended to incorporate any prior belief
directly on the portfolio weights. Linear shrinkage solutions optimally
combining different portfolio allocation rules, such as the GMVP, the
portfolio with equal weights and the tangency portfolio (cf. Section
\ref{secDataFor}) have been reported in \cite{KZ07,GUW07,FMW10,TZ11}.

Alternative approaches have been based on resampling techniques
\cite{M93B,BLW09b}, as well as stochastic programming and also robust
estimation, where the emphasis is on robustifying estimators that are
efficient under the assumption of Gaussian asset-returns, and which are
usually highly sensitive to deviations from the distributional assumption
(see, e.g., \cite{DN09} and references therein). Finally, a line of
contributions from statistical physics initiated by \cite{PGRAS99,LCBPa99}
have been reporting on a methodology based on random matrix theory that
consists of preserving the stability over time of the covariance matrix
estimator by filtering noisy eigenvalues conveying no valuable information.
The cleaning mechanism relies on the empirical fact that relevant information
is structurally captured by some few eigenvalues, while the rest can be
ascribed to noise and measurement errors and resemble the spectrum of a white
covariance matrix (see also \cite{BP11}).

In this paper, we are interested in the class of structured portfolio
estimators based on the combination of Bayesian or James-Stein shrinkage and
sample weighting. Motivated by the widespread application of this class of
statistical methods in the practice of portfolio and risk management, our
focus is on the performance of portfolio constructions as a function of the
set of weights as well as the shrinkage targets and intensity coefficients
parameterizing the improved moment forecasts. The extension of the statistical
performance analysis of sample optimal portfolios with standard moment
estimates to the case of improved shrinkage estimators is not straightforward.
We concentrate on the consistency analysis by considering a limiting regime
that is defined by both the number of samples and the portfolio dimension
going to infinity at the same rate. Such an asymptotic setting will prove to
be more convenient to characterize realistic, finite-dimensional practical
conditions, where sample-size and number of assets are comparable in
magnitude. In particular, we resort to some recent results from the theory of
the spectral analysis of large random matrices, which as in \cite{BLW09b} and
contrary to the random matrix theoretical contributions from statistical
physics cited above, are based on Stieltjes transform methods and stochastic
convergence theory.

Before outlining the contributions and structure of the work, we draw some
connections between the subject of the paper and classical methods in the
statistical signal processing literature. As a matter of fact, (\ref{DGP})
encompasses a broad range of system configurations described by the general
vector channel model. In fact, as for the mean-variance portfolio optimization
problem, usual linear filtering schemes solving typical signal waveform
estimation and detection problems in sensor array processing and wireless
communications are based on the estimation of the unknown observation
covariance matrix as well as possibly a vector of cross-correlations with a
pilot training sequence. Prominent examples are the Capon or minimum variance
spatial filter as well as the minimum mean-square error beamformer and
detector \cite{V98B,VT02B}, and also adaptive filtering\footnote{In
particular, typical formulations of this problem based on (weighted)
least-squares regression are intimately related to the passive investment
strategy of index tracking (see, e.g., \cite[Chapter 4]{P07B}).} applications
\cite{XH05}, in all of which both Bayesian and regularization (shrinkage)
methods are widely applied. Indeed, robust methods are similarly well-known
and extensively used in signal processing applications (see examples in, e.g.,
\cite{LS05B,PE10B}). In particular, norm-constrains have been extensively
investigated in the sensor array signal processing literature (see, e.g.,
\cite{LSW04}). Finally, analyses of weighted sample estimators of covariance
matrices can be found in \cite{Z03,BBT08} and applications of the bootstrap in
\cite{ZB98}.

\subsection{Contributions and structure of the work}

The main contributions of the paper are as follows. We first characterize the
consistency of sample mean-variance portfolios based on the aforementioned
improved moment estimators by providing asymptotic deterministic equivalents
of the achieved out-of-sample performance in the more meaningful double-limit
regime introduced above. Our analytical framework allows us to quantify and
better understand the impact of estimation errors on the out-of-the-sample
performance of optimal portfolios. Specifically, we provide a precise
quantitative description of the amount of risk underestimation and return
overestimation of portfolio constructions based on improved estimators, in a
way depending on the ratio of the portfolio dimension to sample-size as well
as the underlying investment scenario. This phenomena, which render overly
optimistic any investment assessment and decision based on estimated Sharpe
ratios, has already been observed in the financial literature for standard
portfolio implementations.

Furthermore, we propose a class of mean-variance portfolio estimators defined
in terms of a set of weights and shrinkage parameters calibrated so as to
optimize the achieved out-of-sample performance. In essence, an optimal
parameterization is obtained by effectively correcting the analytically
derived asymptotic deviations of the performance of sample portfolios.

The structure of the work is as follows. After the brief literature account
and introductory research motivations in this section, Section
\ref{secDataFor} introduces the modeling details and the moment forecasting
schemes considered in this paper. The problem of evaluating the out-of-sample
performance of large portfolios is also explained. In Section \ref{secADE}, we
provide a characterization of the performance of improved estimators based on
sample weighting and James-Stein shrinkage. Observed deviations from optimal
performance are corrected in Section \ref{secGCE}, where we propose a class of
improved portfolios for high-dimensional settings. Section \ref{secSim}
presents some simulation work validating our theoretical findings and Section
\ref{secCon} concludes the contribution by summarizing the paper. Technical
results and proofs are relegated to the appendices.

\section{Data model and problem formulation%
\label{secDataFor}%
}

Consider the time series with the logarithmic differences of the prices of $M$
financial assets at the edges of an investment period with time-horizon $t$.
Generally enough, we can define the data generating process of the previous
compound or log returns by the following vector stochastic
process\footnote{\textbf{Notation}: All vectors are defined as column vectors
and designated with bold lower case; all matrices are given in bold upper
case; for both vectors and matrices a subscript will be added to emphasize
dependence on dimension, though it will be occasionally dropped for the sake
of clarity of presentation; $\left[  \cdot\right]  _{j}$ will be used for the
$j$th entry of a vector; $\left(  \cdot\right)  ^{T}$ denotes transpose;
$\mathbf{I}_{M}$ denotes the $M\times M$ identity matrix; $\mathbf{1}_{M}$
denotes an $M$ dimensional vector with all entries equal to one;
$\operatorname*{tr}\left[  \cdot\right]  $ denotes the matrix trace operator;
$\mathbb{R}$ and $\mathbb{C}$ denote the real and complex fields of dimension
specified by a superscript; $\operatorname{Im}\left\{  z\right\}  $ denotes
imaginary part of the complex argument; $\mathbb{R}^{\mathbb{+}}=\left\{
z\in\mathbb{C}:\operatorname{Im}\left\{  z\right\}  >0\right\}  $;
$\mathbb{C}^{\mathbb{+}}=\left\{  z\in\mathbb{C}:\operatorname{Im}\left\{
z\right\}  >0\right\}  $; $\mathbb{E}\left[  \cdot\right]  $ denotes
expectation; given two quantites $a,b$, $a\asymp b$ will denote both
quantities are asymptotic equivalents, i.e., $\left\vert a-b\right\vert
\overset{a.s.}{\rightarrow}0$, with $a.s.$ denoting almost sure convergence;
$K,K_{p}$ denote constant values not depending on any relevant quantity, apart
from the latter on a parameter $p$; $\left\vert \cdot\right\vert $ denotes
absolute value and $\left\Vert \cdot\right\Vert $ denotes the Euclidean norm
for vectors and the induced norm for matrices (i.e., spectral or strong norm),
whereas $\left\Vert \cdot\right\Vert _{F}$ denotes Frobenius norm, i.e., for a
matrix $\mathbf{A}\in\mathbb{C}^{M\times M}$ with eigenvalues denoted by
$\lambda_{m}\left(  \mathbf{A}\right)  $, $m=1,\ldots,M$, such that
$\lambda_{M}\left(  \mathbf{A}\right)  \leq\lambda_{M-1}\left(  \mathbf{A}%
\right)  \leq\ldots\leq\lambda_{1}\left(  \mathbf{A}\right)  $, and spectral
radius $\rho\left(  \mathbf{A}\right)  =\max_{1\leq m\leq M}\left(  \left\vert
\lambda_{m}\right\vert \right)  $, $\left\Vert \mathbf{A}\right\Vert =\left(
\rho\left(  \mathbf{A}^{H}\mathbf{A}\right)  \right)  ^{1/2}$, $\left\Vert
\mathbf{A}\right\Vert _{F}=\left(  \operatorname*{Tr}\left[  \mathbf{A}%
^{H}\mathbf{A}\right]  \right)  ^{1/2}$, $\left\Vert \mathbf{A}\right\Vert
_{\operatorname*{tr}}=\operatorname*{Tr}\left[  \left(  \mathbf{A}%
^{H}\mathbf{A}\right)  ^{1/2}\right]  $.}:%
\begin{equation}
\mathbf{y}_{t}=\boldsymbol{\mu}_{t}+\boldsymbol{\varepsilon}_{t}\text{,\quad
}\boldsymbol{\varepsilon}_{t}=\mathbf{\Sigma}_{t}^{1/2}\mathbf{x}_{t}\text{,}
\label{DGP}%
\end{equation}
where $\boldsymbol{\mu}_{t}$ and $\mathbf{\Sigma}_{t}$ are the expected value
and covariance matrix of the asset returns over the investment period, and
$\mathbf{x}_{t}$ is a random vector with independent and identically
distributed (i.i.d.) entries having mean zero and variance one. We are
interested in the problem of optimal single-period (static) mean-variance
portfolio selection, which can be mathematically formulated as the following
quadratic optimization problem with linear constraints:%
\begin{equation}%
\begin{array}
[c]{rl}%
\underset{\mathbf{w}_{t}}{\min} & \mathbf{w}_{t}^{T}\mathbf{\Sigma}%
_{t}\mathbf{w}_{t}\\
\mathsf{s.t.} & \mathbf{w}_{t}^{T}\boldsymbol{\mu}_{t}=\mu_{d}\\
& \mathbf{w}_{t}^{T}\mathbf{1}_{M}=1\text{,}%
\end{array}
\label{MVPO}%
\end{equation}
where $\mu_{d}$ represents the target or desired level of expected portfolio
return%
\footnote
{As conventionally, and for the sake of clarity of presentation, we will assume that logarithmic returns are well approximated by their linear counterparts, so that we can claim about the additivity of returns over both portfolio assets and intertemporally.}%
, and $\mathbf{w}_{t}^{T}\mathbf{1}_{M}=1$\ is a budget constraint.

We shall assume without loss of generality that the forecasting sampling
frequency coincides with the rebalancing frequency. In particular, mean vector
and covariance matrix are forecasted with the return data over a prescribed
estimation window up to the time of the investment decision. Since we only
consider the case of a single-period investment horizon, in the sequel we will
omit the subscript and let $\mathbf{w}_{t}=\mathbf{w}$ for notational
convenience. The solution to (\ref{MVPO}) is straightforwardly given by%

\begin{equation}
\mathbf{w}_{\mathsf{MV}}=\frac{C-\mu_{d}B}{AC-B^{2}}\mathbf{\Sigma}_{t}%
^{-1}\mathbf{1}_{M}+\frac{\mu_{d}A-B}{AC-B^{2}}\mathbf{\Sigma}_{t}%
^{-1}\boldsymbol{\mu}_{t}\text{,} \label{MVP}%
\end{equation}
where $A=\mathbf{1}_{M}^{T}\mathbf{\Sigma}_{t}^{-1}\mathbf{1}_{M}$,
$B=\mathbf{1}_{M}^{T}\mathbf{\Sigma}_{t}^{-1}\boldsymbol{\mu}_{t}$ and
$C=\boldsymbol{\mu}_{t}{}^{T}\mathbf{\Sigma}_{t}^{-1}\boldsymbol{\mu}_{t}$. In
particular, if the constraint on the level of return achieved is dropped, then
we obtain the so-called \textit{global minimum variance portfolio} (GMVP),
which is given by%
\begin{equation}
\mathbf{w}_{\mathsf{GMVP}}=\frac{\mathbf{\Sigma}_{t}^{-1}\mathbf{1}_{M}%
}{\mathbf{1}_{M}^{T}\mathbf{\Sigma}_{t}^{-1}\mathbf{1}_{M}}\text{.}
\label{GMVP}%
\end{equation}
In fact, the latter is clearly also the solution to the general mean-variance
problem if $\mathbf{\mu}_{t}=0$, as it is often assumed over short investment
periods. Other special case of particular interest due to its implications in
asset pricing theory is that of the \textit{tangency portfolio} (TP), which is
given by%
\begin{equation}
\mathbf{w}_{\mathsf{TP}}=\frac{\mathbf{\Sigma}_{t}^{-1}\boldsymbol{\mu}_{t}%
}{\mathbf{1}_{M}^{T}\mathbf{\Sigma}_{t}^{-1}\boldsymbol{\mu}_{t}}\text{.}
\label{TP}%
\end{equation}

In practice, $\boldsymbol{\mu}_{t}$ and $\mathbf{\Sigma}_{t}$ are unknown and
so they must be estimated from market data observations. Let $\boldsymbol{\hat
{\mu}}_{t}$ and $\mathbf{\hat{\Sigma}}_{t}$ denote the forecasted values of
the expected mean and the covariance matrix, respectively. Moreover, let
$\mathbf{\hat{w}}_{\mathsf{GMVP}}$ and $\mathbf{\hat{w}}_{\mathsf{TP}}$
represent the sample construction of (\ref{GMVP}) and (\ref{TP}),
respectively, based on the previous moment estimates. In the following, we
briefly elaborate on the classical forecasting settings that are customarily
applied to estimate the input parameters of the Markowitz portfolio
optimization framework. Specifically, we consider in the first place the
conventional assumption according to which the returns over consecutive
investment periods are independent and identically distributed, and the two
required moments are obtained by their respective unconditional estimators.
Then, we turn our attention to conditional forecasting models based on linear
and stationary stochastic processes; finally, we shortly comment on
heteroscedastic models allowing for some time-variability of the multivariate
volatility process.

Before proceeding further, let us introduce some useful notation. We will
denote by $\left\{  \mathcal{F}_{t-1}\right\}  $ the information set of events
up to the discrete-time instant $t-1$, i.e., the $\sigma$-field generated by
the observed series $\left\{  \mathbf{y}_{l}\right\}  _{l<t}$. Conditional on
the observation available up to the investment decision time, the covariance
matrix of the stochastic process $\mathbf{y}_{t}$ is given by definition by
$\mathbf{\Sigma}_{t}=\operatorname*{var}\left(  \left.  \mathbf{y}%
_{t}\right\vert \mathcal{F}_{t-1}\right)  =\operatorname*{var}\left(  \left.
\boldsymbol{\varepsilon}_{t}\right\vert \mathcal{F}_{t-1}\right)  $.
Additionally, we let $\mathbf{Y}_{N}=\left[  \mathbf{y}_{t-N},\ldots
,\mathbf{y}_{t-1}\right]  $ denote the sample data matrix with the $N$ past
return observations.

\subsection{The case of IID returns: weighted sampling and shrinkage
estimation%
\label{ssecUE}%
}

Under the classical assumption of i.i.d. returns, mean vector and covariance
matrix are both modeled as constant over the entire estimation interval, i.e.,
$\boldsymbol{\mu}_{l}=\boldsymbol{\mu}$, $\mathbf{\Sigma}_{l}=\mathbf{\Sigma}%
$, $l=t-N,\ldots,t-1$. Hence, the standard forecasts of the moments are given
in terms of a rolling-window by the (unconditional) sample mean and sample
covariance matrix, i.e., respectively,%
\begin{equation}
\boldsymbol{\hat{\mu}}=\frac{1}{N}%
{\displaystyle\sum\limits_{n=t-N}^{t-1}}
\mathbf{y}_{n}=\frac{1}{N}\mathbf{Y}_{N}\mathbf{1}_{N}\text{,} \label{SM}%
\end{equation}
and%
\begin{equation}
\mathbf{\hat{\Sigma}}=\frac{1}{N}%
{\displaystyle\sum\limits_{n=t-N}^{t-1}}
\left(  \mathbf{y}_{n}-\mathbf{\hat{\mu}}\right)  \left(  \mathbf{y}%
_{n}-\mathbf{\hat{\mu}}\right)  ^{T}=\frac{1}{N}\mathbf{Y}_{N}\left(
\mathbf{I}_{N}-\frac{1}{N}\mathbf{1}_{N}\mathbf{1}_{N}^{T}\right)
\mathbf{Y}_{N}^{T}\text{.} \label{SCM}%
\end{equation}

A classical extension of the standard estimators in (\ref{SM}) and (\ref{SCM})
considers the effect of weighting the sample observations. Let $\mathbf{W}%
_{\boldsymbol{\mu},N}\in\mathbb{R}^{N\times N}$ and $\mathbf{\mathbf{W}%
}_{\mathbf{\Sigma},N}\in\mathbb{R}^{N\times N}$ be two diagonal matrices with
entries given by a set of nonnegative coefficients, respectively,
$w_{\boldsymbol{\mu},n}$ and $w_{\mathbf{\Sigma},n}$, $n=1,\ldots,N$.
Specifically, the weighted sample mean and weighted sample covariance matrix
are respectively defined as%
\begin{equation}
\boldsymbol{\hat{\mu}}_{\mathsf{W}}=\frac{1}{N}%
{\displaystyle\sum\limits_{n=t-N}^{t-1}}
w_{\boldsymbol{\mu},n}\mathbf{y}_{n}=\frac{1}{N}\mathbf{Y}_{N}%
\mathbf{\mathbf{W}}_{\boldsymbol{\mu},N}\mathbf{1}_{N}\text{.} \label{WSM}%
\end{equation}
and%
\begin{align}
\mathbf{\hat{\Sigma}}_{\mathsf{W}}  &  =\frac{1}{N}%
{\displaystyle\sum\limits_{n=t-N}^{t-1}}
w_{\mathbf{\Sigma},n}\left(  \mathbf{y}_{n}-\boldsymbol{\hat{\mu}}%
_{\mathsf{W}}\right)  \left(  \mathbf{y}_{n}-\boldsymbol{\hat{\mu}%
}_{\mathsf{W}}\right)  ^{T}\nonumber\\
&  =\frac{1}{N}\mathbf{Y}_{N}\left(  \mathbf{I}_{N}-\frac{1}{N}%
\mathbf{\mathbf{W}}_{\boldsymbol{\mu},N}\mathbf{1}_{N}\mathbf{1}_{N}%
^{T}\right)  \mathbf{W}_{\mathbf{\Sigma},N}\left(  \mathbf{I}_{N}-\frac{1}%
{N}\mathbf{1}_{N}\mathbf{1}_{N}^{T}\mathbf{\mathbf{W}}_{\boldsymbol{\mu}%
,N}\right)  \mathbf{Y}_{N}^{T}\text{.} \label{WSCM}%
\end{align}
Weighted estimators are usually applied in order to reduce variability and
improve the stability of parameter estimators, for instance by using
stratified random sampling \cite{S03B}. A related structure is the one
obtained by the nonparametric bootstrap, for which the weights represent the
number of times the corresponding observation appears in the bootstrap sample
\cite{ET93B}\footnote{We assume that the choice of weights is given; possible
weighting schemes range from the standard simple random sampling with
replacement (i.e., uniform resampling following a multinomial distribution) to
sampling from the empirical distribution of the asset returns with nonuniform
weights by for instance assigning different resampling probabilities to the
different observations using importance sampling (see, e.g., \cite{H92B,VW96B}
for more details)}. In the context of asset allocation, \cite{M89} (see also
\cite{M93B}) suggests averaging a sequence of portfolios obtained by
resampling with replacement from the originally available sample. Regarded as
bootstrap aggregating of \textit{bagging}, such averages are used in
statistics for variance reduction purposes as well as to stabilize the
prediction out-of-sample performance as a remedy to overfitting (see, e.g.,
Chapter 10 in \cite{R04B}). In particular, the bootstrap is typically used to
provide small-sample corrections for possibly consistent but biased
estimators. However, in high-dimensional settings, the standard application of
the bootstrap generally yields inconsistent estimates of bias. An asymptotic
refinement of the conventional bootstrap-based bias correction (see, e.g.,
\cite{H01} for standard methodology) is provided in \cite{BLW09b} by resorting
to random matrix theoretical results.

A common further extension to (possibly weighted) sample estimation relies on
the widespread family of Steinian (James-Stein-type) shrinkage estimators of
the mean and covariance matrix of the observed samples. By means of
regularizing or shrinking the estimators (\ref{WSM}) and (\ref{WSCM}), we
define:%
\begin{equation}
\boldsymbol{\hat{\mu}}_{\mathsf{SHR}}=\left(  1-\delta\right)
\boldsymbol{\hat{\mu}}_{\mathsf{W}}+\delta\boldsymbol{\mu}_{0}\text{,}%
\label{meanSHR}%
\end{equation}
and%
\begin{equation}
\mathbf{\hat{\Sigma}}_{\mathsf{SHR}}=\left(  1-\rho\right)  \mathbf{\hat
{\Sigma}}_{\mathsf{W}}+\rho\mathbf{\Sigma}_{0}\text{,}\label{covSHR}%
\end{equation}
where the nonrandom vector $\boldsymbol{\mu}_{0}$ and the positive matrix
$\mathbf{\Sigma}_{0}$ are the shrinkage targets or, from a Bayesian
perspective, the prior knowledge about the unknown $\boldsymbol{\mu}$ and
$\mathbf{\Sigma}$, respectively,\ where $\delta$ are $\rho$ are the shrinkage
intensity parameters. Clearly, if the shrinkage intensity parameters are equal
to $1$ and $\mathbf{W}_{\boldsymbol{\mu},N}=\mathbf{\mathbf{W}}%
_{\mathbf{\Sigma},N}=\mathbf{I}_{N}$, then the standard sample estimators are
recovered. A typical example of shrinkage target for the covariance estimation
is $\mathbf{\Sigma}_{0}=\mathbf{I}_{M}$. Shinkage estimators in the context of
portfolio optimization were first proposed in \cite{J82} (see also \cite{LW03}
for the covariance matrix, and \cite{W07} for a study of the combination of
resampling and shrinkage).

As mentioned in the introduction, it has been recognized in the financial
literature that, under severe estimation risk conditions, the estimated
Markowitz's optimal portfolio rule and its various sophisticated extensions
underperform out-of-the-sample the naive rule based on the \textit{equally
weighted portfolio} (EWP) choice. In an effort to incorporate this well-known
fact into the portfolio selection process, some authors have considered
optimizing a combination\footnote{The rational behind this approach lies on
the so-called fund-separation theorems in finance (see \cite{L97B}).} of one
or more sample portfolios, such as $\mathbf{\hat{w}}_{\mathsf{GMVP}}$ and
$\mathbf{\hat{w}}_{\mathsf{TP}}$, and the uniformly weighted asset allocation
given by $\mathbf{w}_{\mathsf{EWP}}=\mathbf{1}_{M}/M$ (see
\cite{DGU09,KZ07,TZ11}).

\subsection{Accounting for serial dependence: conditional models%
\label{ssecCE}%
}

The previous unconditional estimators of the moments of the asset returns are
particularly well-suited for situations of static nature. Under a more general
setting challenging the i.i.d. assumption, although a period-by-period
computation of the sample statistics by means of a rolling window can indeed
allow for some return predictability, the dynamic behavior of the input
parameters is best modeled in practice by taking into account conditional
information. For the sake of a more precise motivation, we first recall some
empirically observed properties or attributes of time series of asset returns,
the so-called stylized facts in the theory and practice of finance (see
\cite{C01}, and also \cite{MFE05B} for a textbook exposition)%
\footnote
{Although these facts approximately hold unchanged for different time intervals, some characteristics might arguably vary depending on the sampling frequency. According to the time elapsed between return observations, one might differentiate among long-term returns (e.g., weekly, monthly or yearly returns) and short-term returns (i.e., daily returns) - we have omitted purposely a further category including high-frequency data (i.e., intraday, tick data), as it requires different statistical methodologies which we will not consider in this work.}%
. Concerning their distributional properties, it has been observed that return
series are leptokurtic or heavy-tailed (except for long time intervals, for
which the log-normal assumption seems reasonable, at least for
well-diversified portfolios), and extreme return values usually appear in
clusters. Regarding their dynamics, conditional expected returns are usually
negligible (at least relative to volatility values), and, more importantly,
are not independent though exhibit little serial correlation. Conversely,
squared returns, which are often used as a proxy of the unobserved covariance,
show profound evidence of positive serial correlation with high persistence.

If we set aside the time variability of conditional covariances (i.e.,
particularly for long-term horizons), the dynamic dependence structure of the
asset returns can be captured irrespectively of whether its origin is
momentum, mean-reversion, or lead-lag relations by conditionally modeling the
mean via a vector auto-regressive moving-average process (VARMA) with both
orders equal one:%
\begin{equation}
\mathbf{y}_{t}=\boldsymbol{\bar{\mu}}+\mathbf{\Phi}_{M}\mathbf{y}%
_{t-i}+\boldsymbol{\varepsilon}_{t}-\mathbf{\Pi}_{M}\boldsymbol{\varepsilon
}_{t-j}\text{,} \label{VARMAsr}%
\end{equation}
where $\mathbf{\Phi}_{M}$ and $\mathbf{\Pi}_{M}$ are square fixed parameter
matrices, and $\boldsymbol{\varepsilon}_{t}=\mathbf{\Sigma}^{1/2}%
\mathbf{x}_{t}$. The process is customarily assumed to be weakly (second-order
or covariance) stationary and ergodic, as well as stable and invertible (see
\cite{BD91B} for detailed characterization of multivariate time series
models). VARMA processes of higher orders than the VARMA$\left(  1,1\right)  $
in (\ref{VARMAsr}) are reported in the literature to be of less practical
interest \cite{HSW09}, and even further restrictions leading to first-order
vector autoregressions (i.e., $\mathbf{\Theta}_{M}=\mathbf{0}$) are most often
considered (see \cite{CV02B}, and the more recent account in \cite{DNU10}). In
the large dimensional portfolio setting, parsimony is crucial to maintain the
efficiency and low complexity of the model estimation process, and so
different simplifications based on structural restrictions are usually
considered in practice. In particular, in the case of processes with scalar
parameter matrices $\mathbf{\Phi}_{M}=\phi\mathbf{I}_{M}$ and $\mathbf{\Pi
}_{M}=\pi\mathbf{I}_{M}$, then the population covariance matrix has a
particular sparse structure, separable into cross-sectional $\mathbf{\Sigma}$
and temporal covariance components, which we will denote by $\mathbf{\Omega}$.
Under the Gaussian assumption, the sample covariance matrix of members of this
class of VARMA processes are doubly-correlated Wishart matrices. This matrix
ensemble has been recently analyzed in the statistical physics literature in
the context of financial applications \cite{BJJNPZ10}.

According to the VARMA model, the conditional covariance remains constant
regardless of the data. Especially for short-term horizons (e.g., daily
returns), the observed features of the volatility process are best accounted
for by conditional heteroscedastic models, such as the class of specifications
for the multivariate extension of generalized autoregressive conditionally
heteroscedastic (GARCH) process \cite{BLR06}, and the exponential weighted
moving average (EWMA) scheme:%
\begin{equation}
\mathbf{\Sigma}_{t}=\lambda\mathbf{y}_{t-1}\mathbf{y}_{t-1}^{T}+\left(
1-\lambda\right)  \mathbf{\Sigma}_{t-1}=\lambda%
{\displaystyle\sum\limits_{n=t-N}^{t-1}}
\left(  1-\lambda\right)  ^{n-1}\mathbf{y}_{n}\mathbf{y}_{n}^{T}\text{,}
\label{EWMA}%
\end{equation}
where $\lambda$\ is a smoothing prescribed parameter that characterizes the
decay of the exponential memory%
\footnote
{Other possible and common choices of the weights for the past returns (adding up to 1) are equal weights (i.e. a rectangular window with equal weights), exponential weights (i.e. equivalent to an exponential moving average), weights following a power-law decay, or long memory weights (decaying logarithmically slowly).}%
. Firstly proposed in \cite{JPM06}, the EWMA model has been found very useful
in estimating the market risk of portfolios, as well as in portfolio
optimization \cite{LW98}.

\subsection{Evaluating the performance of sample portfolios%
\label{ssecPE}%
}

The quality of a portfolio rule $\mathbf{\hat{w}}$ constructed based on
\textit{in-sample} forecasts of $\boldsymbol{\mu}_{t}$ and $\mathbf{\Sigma
}_{t}$ can be measured by its achieved out-of-sample (realized) mean return
$\mu_{P}\left(  \mathbf{\hat{w}}\right)  =\mathbf{\hat{w}}^{T}\boldsymbol{\mu
}_{t}$ and risk $\sigma_{P}\left(  \mathbf{\hat{w}}\right)  =\sqrt
{\mathbf{\hat{w}}^{T}\mathbf{\Sigma}_{t}\mathbf{\hat{w}}}$. In the study and
practice of finance, measures of risk-adjusted achieved return are usually
employed, being the Sharpe ratio a prominent one in portfolio management:%
\begin{equation}
\mathsf{SR}\left(  \mathbf{\hat{w}}\right)  =\frac{\mu_{P}\left(
\mathbf{\hat{w}}\right)  }{\sigma_{P}\left(  \mathbf{\hat{w}}\right)
}\text{.} \label{sharpe}%
\end{equation}
In particular, notice that the tangency portfolio defined in (\ref{TP}) is the
portfolio that maximizes (\ref{sharpe}) under the budget constraint.

As discussed in the introduction, for a small sample-size and relatively large
universe of assets, the out-of-sample performance of standard portfolio
constructions can be expected to considerably differ from the theoretical
performance given by the true moments. In this paper, we extend existing
analyses of the statistical properties of portfolio rules based on the
standard sample mean and sample covariance matrix estimators, and characterize
the performance deviations due to estimation risk in terms of nonrandom model
and scenario parameters. We will concentrate on the case of the unconditional
moment estimators (\ref{meanSHR}) and (\ref{covSHR}) and conditional VARMA
models with separable covariance structure. Specifically, we derive asymptotic
deterministic equivalents of the out-of-sample performance of improved
portfolio implementations that are based on the previous estimators. We remark
here that usual choices among practitioners of conditional heteroscedastic
models, such as the EWMA model in (\ref{EWMA}), can also be fitted into our
asymptotic framework by resorting to random matrix theoretical results dealing
with general variance profiles\ (see \cite{HLN07}). Furthermore, we provide a
mechanism to calibrate the set of weights and shrinkage parameters defining
the improved portfolio constructions so as to optimize the achieved
out-of-sample performance.

\section{Main results: Out-of-sample analysis and asymptotic corrections}

In this section, we provide the main two results of the paper on the
asymptotic characterization of the performance of sample portfolios and the
proposed family of generalized consistent portfolio estimators are stated in
Section \ref{secADE} and Section \ref{secGCE}, respectively. We first
summarize the technical hypotheses supporting our research and introduce some
new definitions:

\begin{description}
\item[\textbf{(As1)}] Let $\mathbf{R}_{M}\in\mathbb{R}^{M\times M}$ and
$\mathbf{T}_{N}\in\mathbb{R}^{N\times N}$ be two deterministic nonnegative
matrices having spectral norm bounded uniformly in $M$ and $N$, i.e.,
$\left\Vert \mathbf{R}_{M}\right\Vert _{\sup}=\sup_{M\geq1}\left\Vert
\mathbf{R}_{M}\right\Vert <+\infty$ and $\left\Vert \mathbf{T}_{N}\right\Vert
_{\sup}=\sup_{N\geq1}\left\Vert \mathbf{T}_{N}\right\Vert <+\infty$,
respectively; the matrix $\mathbf{T}_{N}$ is diagonal with entries denoted by
$t_{n}$, $1\leq n\leq N$.

\item[\textbf{(As2)}] Let $\mathbf{X}_{M}$ be an $M\times N$ matrix whose
elements $X_{ij}$, $1\leq i\leq M$, $1\leq j\leq N$, are i.i.d. standardized
Gaussian random variables.

\item[\textbf{(As3)}] We will consider the limiting regime defined by both
dimensions $M$ and $N$ growing large without bound at the same rate, i.e.,
$N,M\rightarrow\infty$ such that $0<\lim\inf c_{M}\leq\lim\sup c_{M}<\infty$,
with $c_{M}=M/N$. Quantities that, under the previous double-limit regime, are
asymptotically equivalent to a given a random variable, both depending on $M$
and $N$, will be referred to as \textit{asymptotic deterministic equivalents},
if only depend upon nonrandom model variables, and \textit{generalized
consistent estimators}, if they depend on observable random variables (e.g.,
sample data matrix).
\end{description}

Before proceeding with the out-of-sample performance characterization, we
identify next the key quantities of study into which the Sharpe ratio
performance measure in (\ref{sharpe}) can be decomposed. Let us first consider
the unconditional model, where $\left\{  \boldsymbol{\mu}_{t},\mathbf{\Sigma
}_{t}\right\}  $ are considered to be constant over the estimation window.
Then, notice that the data observation matrix can be written as $\mathbf{Y}%
_{N}=\boldsymbol{\mu}\mathbf{1}_{N}^{T}+\mathbf{\Sigma}^{1/2}\mathbf{X}_{N}$,
where $\mathbf{X}_{N}=\left[  \mathbf{x}_{t-N},\ldots,\mathbf{x}_{t-1}\right]
$. For the sake of clarity of presentation, we will assume that the standard
sample mean in (\ref{SM}) instead of its weighted version is applied in the
definition of $\mathbf{\hat{\Sigma}}_{\mathsf{W}}$. Moreover, we also assume
that the entries of $\mathbf{\mathbf{W}}_{\boldsymbol{\mu},N}$ are chosen to
be the eigenvalues of the matrix $\left(  \mathbf{I}_{N}-\frac{1}{N}%
\mathbf{1}_{N}\mathbf{1}_{N}^{T}\right)  \mathbf{W}_{\mathbf{\Sigma},N}\left(
\mathbf{I}_{N}-\frac{1}{N}\mathbf{1}_{N}\mathbf{1}_{N}^{T}\right)  $, which
will be denoted by $\mathbf{T}_{N}$. In particular, observe that, in this
case,%
\[
\mathbf{\hat{\Sigma}}_{\mathsf{W}}=\frac{1}{N}\mathbf{\Sigma}^{1/2}%
\mathbf{X}_{N}\mathbf{T}_{N}\mathbf{X}_{N}^{T}\mathbf{\Sigma}^{1/2}\text{,}%
\]
where we have used the fact that%
\begin{equation}
\left(  \boldsymbol{\mu}\mathbf{1}_{N}^{T}+\mathbf{\Sigma}^{1/2}\mathbf{X}%
_{N}\right)  \left(  \mathbf{I}_{N}-\frac{1}{N}\mathbf{1}_{N}\mathbf{1}%
_{N}^{T}\right)  =\mathbf{\Sigma}^{1/2}\mathbf{X}_{N}\left(  \mathbf{I}%
_{N}-\frac{1}{N}\mathbf{1}_{N}\mathbf{1}_{N}^{T}\right)  \text{,}
\label{matrix}%
\end{equation}
along with the invariance of the multivariate Gaussian distribution to
orthogonal transformations. Notice that, under the Gaussian assumption, the
matrix in \ref{matrix} is matrix-variate normal distributed, i.e.,
$\mathbf{\Sigma}^{1/2}\mathbf{X}_{N}\left(  \mathbf{I}_{N}-\frac{1}%
{N}\mathbf{1}_{N}\mathbf{1}_{N}^{T}\right)  \sim\mathcal{MN}_{M\times
N}\left(  \mathbf{0}_{M\times N},\mathbf{\Sigma},\mathbf{I}_{N}\right)  $, or
equivalently, $\operatorname*{vec}\left(  \mathbf{\Sigma}^{1/2}\mathbf{X}%
_{N}\left(  \mathbf{I}_{N}-\frac{1}{N}\mathbf{1}_{N}\mathbf{1}_{N}^{T}\right)
\right)  \sim\mathcal{N}_{MN}\left(  \mathbf{0}_{MN},\mathbf{\Sigma}%
\otimes\mathbf{I}_{N}\right)  $; see \cite[Section 3.3.2]{GN00B}.
Consequently, $\mathbf{\hat{\Sigma}}_{\mathsf{W}}$ is a central quadratic
forms (central Wishart distributed if $\mathbf{T}_{N}=\mathbf{I}_{N}$).

Now, let $\mathbf{R}_{M}=\mathbf{\Sigma}_{0}^{-1/2}\mathbf{\Sigma\Sigma}%
_{0}^{-1/2}$, and also, with some abuse of notation, $\mathbf{R}_{M}%
^{1/2}=\mathbf{\Sigma}_{0}^{-1/2}\mathbf{\Sigma}^{1/2}$, and consider further
the nonnegative scalars $\alpha_{M}=\rho/\left(  1-\rho\right)  $ and
$\beta_{M}=\delta/\left(  1-\delta\right)  $. Moreover, we define
$\mathbf{\tilde{Y}}_{N}=\mathbf{R}_{M}^{1/2}\mathbf{X}_{N}\mathbf{T}_{N}%
^{1/2}$, along with $\mathbf{\hat{\Sigma}}_{M}=\tfrac{1}{N}\mathbf{\tilde{Y}%
}_{N}\mathbf{\tilde{Y}}_{N}^{T}+\alpha_{M}\mathbf{I}_{M}$, and
$\boldsymbol{\hat{\upsilon}}_{M}=\frac{1}{N}\mathbf{\tilde{Y}}_{N}%
\boldsymbol{\tilde{\upsilon}}_{N}$, where $\boldsymbol{\tilde{\upsilon}}%
_{N}=\sqrt{N}\boldsymbol{\upsilon}_{N}$, with $\boldsymbol{\upsilon}_{N}$
being an $N$ dimensional nonrandom vector with unit norm. Then, it is
straightforward to see that the numerator and denominator of (\ref{sharpe})
can be written for the class of sample implementations of the optimal
portfolio in (\ref{MVP}) based on the unconditional estimators (\ref{meanSHR})
and (\ref{covSHR}) in terms of the following random variables:%
\begin{align}
\hat{\xi}_{M}^{\left(  1\right)  } &  =\boldsymbol{\upsilon}_{M}%
^{T}\mathbf{\hat{\Sigma}}_{M}^{-1}\boldsymbol{\upsilon}_{M}\text{,}%
\label{xi1}\\
\hat{\xi}_{M}^{\left(  2\right)  } &  =\boldsymbol{\upsilon}_{M}%
^{T}\mathbf{\hat{\Sigma}}_{M}^{-1}\boldsymbol{\hat{\upsilon}}_{M}%
\text{,}\label{xi2}\\
\hat{\xi}_{M}^{\left(  3\right)  } &  =\boldsymbol{\hat{\upsilon}}_{M}%
^{T}\mathbf{\hat{\Sigma}}_{M}^{-1}\boldsymbol{\hat{\upsilon}}_{M}%
\text{,}\label{xi3}\\
\hat{\xi}_{M}^{\left(  4\right)  } &  =\boldsymbol{\upsilon}_{M}%
^{T}\mathbf{\hat{\Sigma}}_{M}^{-1}\mathbf{R}_{M}\mathbf{\hat{\Sigma}}_{M}%
^{-1}\boldsymbol{\upsilon}_{M}\text{,}\label{xi4}\\
\hat{\xi}_{M}^{\left(  5\right)  } &  =\boldsymbol{\upsilon}_{M}%
^{T}\mathbf{\hat{\Sigma}}_{M}^{-1}\mathbf{R}_{M}\mathbf{\hat{\Sigma}}_{M}%
^{-1}\boldsymbol{\hat{\upsilon}}_{M}\text{,}\label{xi5}\\
\hat{\xi}_{M}^{\left(  6\right)  } &  =\boldsymbol{\hat{\upsilon}}_{M}%
^{T}\mathbf{\hat{\Sigma}}_{M}^{-1}\mathbf{R}_{M}\mathbf{\hat{\Sigma}}_{M}%
^{-1}\boldsymbol{\hat{\upsilon}}_{M}\text{.}\label{xi6}%
\end{align}
Notice that similar reasoning applies to the conditional model for the asset
return in (\ref{VARMAsr}), since the sample covariance matrix of the process
is a doubly-correlated Wishart matrix (cf. Section \ref{ssecCE}), and the
conditional mean estimator can be written from the VAR$\left(  1\right)  $
model specification as $\boldsymbol{\hat{\mu}}_{t}=\boldsymbol{\tilde{\mu}%
}+\mathbf{\Sigma}^{1/2}\mathbf{X}_{N}\mathbf{1}_{\Lambda,N}$, where
$\boldsymbol{\tilde{\mu}}=%
{\textstyle\sum\nolimits_{n=1}^{N}}
\phi^{n-1}\boldsymbol{\bar{\mu}}$ and $\mathbf{1}_{\Lambda,N}=\Lambda
_{N}\mathbf{1}_{N}$, with $\Lambda_{N}\in\mathbb{R}^{N\times N}$ being a
diagonal matrix such that $\left[  \Lambda_{N}\right]  _{n}=\phi^{n}$. In
general, the vector $\boldsymbol{\upsilon}_{M}$ takes values in $\left\{
\mathbf{\Sigma}_{0}^{-1/2}\mathbf{1}_{M}/\sqrt{M},\mathbf{\Sigma}_{0}%
^{-1/2}\boldsymbol{\mu},\mathbf{\Sigma}_{0}^{-1/2}\boldsymbol{\mu}%
_{0},\mathbf{\Sigma}_{0}^{-1/2}\boldsymbol{\tilde{\mu}}\right\}  $, and
$\boldsymbol{\upsilon}_{N}$ in $\left\{  \mathbf{1}_{N}/\sqrt{N}%
,\mathbf{1}_{\Lambda,N}\right\}  $.

By way of example, consider the estimation of the quantities $\left\{
A,B,C\right\}  $ defining the optimal mean-variance portfolio in (\ref{MVP}),
based on the unconditional estimators (\ref{meanSHR}) and (\ref{covSHR}) with
$\delta=0$. Let us denote the estimators by $\left\{  \hat{A},\hat{B},\hat
{C}\right\}  $. In particular, observe that $\left(  1-\rho\right)  \hat
{A}=\mathbf{1}_{M}^{T}\mathbf{\Sigma}_{0}^{-1/2}\left(  \mathbf{\Sigma}%
_{0}^{-1/2}\mathbf{\hat{\Sigma}}_{\mathsf{W}}\mathbf{\Sigma}_{0}^{-1/2}%
+\alpha_{M}\mathbf{I}_{M}\right)  ^{-1}\mathbf{\Sigma}_{0}^{-1/2}%
\mathbf{1}_{M}$, and so we readily have $\left(  1-\rho\right)  \hat{A}%
=M\hat{\xi}_{M}^{\left(  1\right)  }$, with $\boldsymbol{\upsilon}%
_{M}=\mathbf{\Sigma}_{0}^{-1/2}\mathbf{1}_{M}/\sqrt{M}$. Moreover, note that
($\mathbf{\mathbf{W}}_{\boldsymbol{\mu},N}\mathbf{1}_{N}=\mathbf{T}_{N}%
^{1/2}\mathbf{\tilde{1}}_{N}$, where $\mathbf{\tilde{1}}_{N}=\mathbf{T}%
_{N}^{-1/2}\mathbf{\mathbf{W}}_{\boldsymbol{\mu},N}\mathbf{1}_{N}$) ($\frac
{1}{N}\mathbf{1}_{N}^{T}\mathbf{\mathbf{W}}_{\boldsymbol{\mu},N}\mathbf{1}%
_{N}=1$ -in definition above-)%
\begin{align*}
\left(  1-\rho\right)  \hat{B}  & =\mathbf{1}_{M}^{T}\mathbf{\Sigma}%
_{0}^{-1/2}\left(  \mathbf{\Sigma}_{0}^{-1/2}\mathbf{\hat{\Sigma}}%
_{\mathsf{W}}\mathbf{\Sigma}_{0}^{-1/2}+\alpha_{M}\mathbf{I}_{M}\right)
^{-1}\mathbf{\Sigma}_{0}^{-1/2}\boldsymbol{\hat{\mu}}_{\mathsf{W}}\\
& =\mathbf{1}_{M}^{T}\mathbf{\Sigma}_{0}^{-1/2}\left(  \mathbf{\Sigma}%
_{0}^{-1/2}\mathbf{\hat{\Sigma}}_{\mathsf{W}}\mathbf{\Sigma}_{0}^{-1/2}%
+\alpha_{M}\mathbf{I}_{M}\right)  ^{-1}\mathbf{\Sigma}_{0}^{-1/2}%
\boldsymbol{\mu}+\frac{1}{N}\mathbf{1}_{M}^{T}\mathbf{\Sigma}_{0}%
^{-1/2}\left(  \mathbf{\Sigma}_{0}^{-1/2}\mathbf{\hat{\Sigma}}_{\mathsf{W}%
}\mathbf{\Sigma}_{0}^{-1/2}+\alpha_{M}\mathbf{I}_{M}\right)  ^{-1}%
\mathbf{\Sigma}_{0}^{-1/2}\mathbf{\Sigma}^{1/2}\mathbf{X}_{N}\mathbf{T}%
_{N}^{1/2}\mathbf{\tilde{1}}_{N}\text{,}%
\end{align*}
and therefore we have $\left(  1-\rho\right)  \hat{B}=\sqrt{M}\hat{\xi}%
_{M}^{\left(  1\right)  }+\sqrt{M}\hat{\xi}_{M}^{\left(  2\right)  }$, where
the vector $\boldsymbol{\upsilon}_{M}$ take values $\boldsymbol{\upsilon}%
_{M}=\mathbf{\Sigma}_{0}^{-1/2}\mathbf{1}_{M}/\sqrt{M}$ and
$\boldsymbol{\upsilon}_{M}=\mathbf{\Sigma}_{0}^{-1/2}\boldsymbol{\mu}$, and
$\boldsymbol{\tilde{\upsilon}}_{N}=\mathbf{\tilde{1}}_{N}$.

Furthermore, for the estimator of $C$, we have that ($\boldsymbol{\hat{\mu}%
}_{\mathsf{W}}=\frac{1}{N}\left(  \boldsymbol{\mu}\mathbf{1}_{N}%
^{T}+\mathbf{\Sigma}^{1/2}\mathbf{X}_{N}\right)  \mathbf{\mathbf{W}%
}_{\boldsymbol{\mu},N}\mathbf{1}_{N}=\frac{1}{N}\boldsymbol{\mu}\mathbf{1}%
_{N}^{T}\mathbf{\mathbf{W}}_{\boldsymbol{\mu},N}\mathbf{1}_{N}+\frac{1}%
{N}\mathbf{\Sigma}^{1/2}\mathbf{X}_{N}\mathbf{\mathbf{W}}_{\boldsymbol{\mu}%
,N}\mathbf{1}_{N}=\boldsymbol{\mu}+\frac{1}{N}\mathbf{\Sigma}^{1/2}%
\mathbf{X}_{N}\mathbf{T}_{N}^{1/2}\mathbf{\tilde{1}}_{N}$) ($\left\Vert
\boldsymbol{\mu}\right\Vert $)%
\begin{align*}
\left(  1-\rho\right)  \hat{C}  & =\boldsymbol{\hat{\mu}}_{\mathsf{W}}%
^{T}\mathbf{\Sigma}_{0}^{-1/2}\left(  \mathbf{\Sigma}_{0}^{-1/2}%
\mathbf{\hat{\Sigma}}_{\mathsf{W}}\mathbf{\Sigma}_{0}^{-1/2}+\alpha
_{M}\mathbf{I}_{M}\right)  ^{-1}\mathbf{\Sigma}_{0}^{-1/2}\boldsymbol{\hat
{\mu}}_{\mathsf{W}}\\
& =\boldsymbol{\mu}^{T}\mathbf{\Sigma}_{0}^{-1/2}\left(  \mathbf{\Sigma}%
_{0}^{-1/2}\mathbf{\hat{\Sigma}}_{\mathsf{W}}\mathbf{\Sigma}_{0}^{-1/2}%
+\alpha_{M}\mathbf{I}_{M}\right)  ^{-1}\mathbf{\Sigma}_{0}^{-1/2}%
\boldsymbol{\mu}\\
& +2\frac{1}{N}\boldsymbol{\mu}^{T}\mathbf{\Sigma}_{0}^{-1/2}\left(
\mathbf{\Sigma}_{0}^{-1/2}\mathbf{\hat{\Sigma}}_{\mathsf{W}}\mathbf{\Sigma
}_{0}^{-1/2}+\alpha_{M}\mathbf{I}_{M}\right)  ^{-1}\mathbf{\Sigma}_{0}%
^{-1/2}\mathbf{\Sigma}^{1/2}\mathbf{X}_{N}\mathbf{T}_{N}^{1/2}\mathbf{\tilde
{1}}_{N}\\
& +\frac{1}{N^{2}}\mathbf{\tilde{1}}_{N}^{T}\mathbf{T}_{N}^{1/2}\mathbf{X}%
_{N}^{T}\mathbf{\Sigma}^{1/2}\mathbf{\Sigma}_{0}^{-1/2}\left(  \mathbf{\Sigma
}_{0}^{-1/2}\mathbf{\hat{\Sigma}}_{\mathsf{W}}\mathbf{\Sigma}_{0}%
^{-1/2}+\alpha_{M}\mathbf{I}_{M}\right)  ^{-1}\mathbf{\Sigma}_{0}%
^{-1/2}\mathbf{\Sigma}^{1/2}\mathbf{X}_{N}\mathbf{T}_{N}^{1/2}\mathbf{\tilde
{1}}_{N}\\
& =\hat{\xi}_{M}^{\left(  1\right)  }+\hat{\xi}_{M}^{\left(  2\right)  }%
+\hat{\xi}_{M}^{\left(  3\right)  }\text{.}%
\end{align*}
Finally, notice that, additionally, the term is required to model the variance
of the GMVP, and so is for modeling the return of the\ TP, but both terms can
be straightforwardly represented similarly as $\left\{  \hat{A},\hat{B}%
,\hat{C}\right\}  $.

$\left(  1-\rho\right)  ^{2}\mathbf{1}_{M}^{T}\mathbf{\hat{\Sigma}%
}_{\mathsf{SHR}}^{-1}\mathbf{\Sigma\hat{\Sigma}}_{\mathsf{SHR}}^{-1}%
\mathbf{1}_{M}=M\hat{\xi}_{M}^{\left(  4\right)  }$,

$\left(  1-\rho\right)  ^{2}\boldsymbol{\hat{\mu}}_{\mathsf{W}}^{T}%
\mathbf{\hat{\Sigma}}_{\mathsf{SHR}}^{-1}\mathbf{\Sigma\hat{\Sigma}%
}_{\mathsf{SHR}}^{-1}\boldsymbol{\hat{\mu}}_{\mathsf{W}}=\hat{\xi}%
_{M}^{\left(  4\right)  }+\hat{\xi}_{M}^{\left(  5\right)  }+\hat{\xi}%
_{M}^{\left(  6\right)  }$.

From above, the previous two assumptions on the weighting matrices clearly
facilitate exposition and tractability. However, we remark that more general
cases can be equivalently reduced to the above key quantities by algebraic
manipulations essentially relying on the matrix inversion lemma (cf. identity
(\ref{MIL}) in Appendix \ref{appTechAux}).

Now that the out-of-sample performance characterization problem has been
reduced to the study of the behavior of the quantities (\ref{xi1}) to
(\ref{xi6}), we proceed in the following two sections with their asymptotic
analysis and consistent estimation.

\subsection{Asymptotic performance analysis: a RMT approach%
\label{secADE}%
}

Define $\gamma=\gamma_{M}=\frac{1}{N}\operatorname*{tr}\left[  \mathbf{E}%
_{M}^{2}\right]  $ and $\tilde{\gamma}=\tilde{\gamma}_{M}=\frac{1}%
{N}\operatorname*{tr}\left[  \mathbf{\tilde{E}}_{N}^{2}\right]  $, where
$\mathbf{E}_{M}=\mathbf{R}_{M}\left(  \tilde{\delta}_{M}\mathbf{R}_{M}%
+\alpha\mathbf{I}_{M}\right)  ^{-1}$ and $\mathbf{\tilde{E}}_{N}%
=\mathbf{T}_{N}\left(  \mathbf{I}_{N}+\delta_{M}\mathbf{T}_{N}\right)  ^{-1}$,
with $\left\{  \tilde{\delta}_{M},\delta_{M}\right\}  $ being the unique
positive solution to the following system of equations \cite[Proposition
1]{HKLNP08}:
\begin{equation}
\left\{
\begin{array}
[c]{l}%
\tilde{\delta}_{M}=\frac{1}{N}\operatorname*{tr}\left[  \mathbf{T}_{N}\left(
\mathbf{I}_{N}+\delta_{M}\mathbf{T}_{N}\right)  ^{-1}\right]  \\
\delta_{M}=\frac{1}{N}\operatorname*{tr}\left[  \mathbf{R}_{M}\left(
\tilde{\delta}_{M}\mathbf{R}_{M}+\alpha\mathbf{I}_{M}\right)  ^{-1}\right]
\text{.}%
\end{array}
\right.  \label{systemEq}%
\end{equation}
Then, we have the following result characterizing the asymptotic behavior of
the random variables (\ref{xi1}) to (\ref{xi6}).%

\begin{theorem}%
\label{theoADE}%
(\textit{Asymptotic Deterministic Equivalents}) Under Assumptions
(\textbf{As1}) to (\textbf{As3}), the following asymptotic equivalences hold
true:%
\begin{align*}
\hat{\xi}_{M}^{\left(  1\right)  } &  \asymp\boldsymbol{\upsilon}_{M}%
^{T}\left(  \tilde{\delta}_{M}\mathbf{R}_{M}+\alpha_{M}\mathbf{I}_{M}\right)
^{-1}\boldsymbol{\upsilon}_{M}\text{,}\\
\hat{\xi}_{M}^{\left(  2\right)  } &  \asymp0\text{,}\\
\hat{\xi}_{M}^{\left(  3\right)  } &  \asymp\delta_{M}\boldsymbol{\upsilon
}_{N}^{T}\mathbf{T}_{N}\left(  \delta_{M}\mathbf{T}_{N}+\mathbf{I}_{N}\right)
^{-1}\boldsymbol{\upsilon}_{N}\text{,}\\
\hat{\xi}_{M}^{\left(  4\right)  } &  \asymp\frac{1}{1-\gamma_{M}\tilde
{\gamma}_{M}}\boldsymbol{\upsilon}_{M}^{T}\mathbf{R}_{M}^{1/2}\left(
\tilde{\delta}_{M}\mathbf{R}_{M}+\alpha_{M}\mathbf{I}_{M}\right)
^{-2}\mathbf{R}_{M}^{1/2}\boldsymbol{\upsilon}_{M}\text{,}\\
\hat{\xi}_{M}^{\left(  5\right)  } &  \asymp0\text{,}\\
\hat{\xi}_{M}^{\left(  6\right)  } &  \asymp\frac{\gamma_{M}}{1-\gamma
_{M}\tilde{\gamma}_{M}}\boldsymbol{\upsilon}_{N}^{T}\mathbf{T}_{N}\left(
\delta_{M}\mathbf{T}_{N}+\mathbf{I}_{N}\right)  ^{-2}\boldsymbol{\upsilon}%
_{N}\text{.}%
\end{align*}%
\end{theorem}%
\begin{proof}%
See Appendix \ref{appTheoADE}.%
\end{proof}%

Using Theorem \ref{theoADE}, estimates of the out-of-sample performance of
optimal sample mean-variance portfolios based on the unconditional and
conditional models in Section \ref{secDataFor} are readily obtained. By means
of the previous asymptotic approximations in a practically more meaningful and
relevant double-limit regime (cf. Section \ref{secSim}), more accurate
information about the underestimation and overestimation effects of the
portfolio risk and return, respectively, can be provided.

The previous result is of interest on its own for characterization purposes as
well as for scenario analysis in investment management. However, particularly
for the calibration of unconditional models, one might well also be interested
in estimates of the previous quantities that are given in terms of the
available information, i.e., essentially, the data observation matrix. In the
proposed asymptotic regime, it follows from Theorem \ref{theoADE} that both
$\hat{\xi}_{M}^{\left(  2\right)  }$ and $\hat{\xi}_{M}^{\left(  5\right)  }$
are negligible and therefore can be discarded for analysis and decision
purposes. While $\hat{\xi}_{M}^{\left(  1\right)  }$ and $\hat{\xi}%
_{M}^{\left(  3\right)  }$ are already given in terms of only observable
data\footnote{This is not the case for $\boldsymbol{\upsilon}_{M}%
=\mathbf{\Sigma}_{0}^{-1/2}\boldsymbol{\mu}$, but still the consistent
estimation of $\hat{\xi}_{M}^{\left(  1\right)  }$ can be handled
straightforwarly by rearranging terms.}, terms $\hat{\xi}_{M}^{\left(
4\right)  }$ and $\hat{\xi}_{M}^{\left(  6\right)  }$ happen to be defined in
terms of the unknown $\mathbf{R}_{M}$. We next present a class of estimators
of (\ref{xi4}) and (\ref{xi6}), or equivalently their asymptotic deterministic
equivalents provided by Theorem \ref{theoADE}, which are strongly consistent
under the limiting regime in (\textbf{As3}).

\subsection{Consistent estimation of optimal large dimensional portfolios%
\label{secGCE}%
}

The parameters defining the estimators in (\ref{meanSHR}) and (\ref{covSHR}),
i.e., $\left\{  \mathbf{W}_{\boldsymbol{\mu},N},\mathbf{\mathbf{W}%
}_{\mathbf{\Sigma},N}\right\}  $ and $\left\{  \delta,\rho\right\}  $,
effectively represent a set of degrees-of-freedom with respect to which the
out-of-sample performance of a portfolio construction can be improved. For the
calibration of unconditional models by means of optimizing the estimator
parameterization, only the available sample data can be used in practice in
order to select the previous set of parameters. To that effect, from the
definition of the quantities (\ref{xi4}) and (\ref{xi6}) and the discussion
above, the estimation of $\hat{\xi}_{M}^{\left(  4\right)  }$ and $\hat{\xi
}_{M}^{\left(  6\right)  }$ are required. The naive approach is based on the
\textit{plug-in} or conventional estimator of $\hat{\xi}_{M}^{\left(
4\right)  }$, henceforth denoted with the subscript "cnv" by $\hat{\xi
}_{cnv,M}^{\left(  4\right)  }$, which is given by replacing the unknown
theoretical covariance matrix by the SCM, i.e.,%
\begin{equation}
\hat{\xi}_{cnv,M}^{\left(  4\right)  }=\boldsymbol{\upsilon}_{M}^{T}\left(
\tfrac{1}{N}\mathbf{\tilde{Y}}_{N}\mathbf{\tilde{Y}}_{N}^{T}+\alpha
_{M}\mathbf{I}_{M}\right)  ^{-1}\tfrac{1}{N}\mathbf{\tilde{Y}}_{N}%
\mathbf{\tilde{Y}}_{N}^{T}\left(  \tfrac{1}{N}\mathbf{\tilde{Y}}%
_{N}\mathbf{\tilde{Y}}_{N}^{T}+\alpha_{M}\mathbf{I}_{M}\right)  ^{-1}%
\boldsymbol{\upsilon}_{M}\text{.} \label{xi4con}%
\end{equation}
Additionally, let $\hat{\xi}_{cnv,M}^{\left(  6\right)  }$ denote the
"plug-in" estimator of $\hat{\xi}_{M}^{\left(  6\right)  }$, and notice that%

\begin{align}
\hat{\xi}_{cnv,M}^{\left(  6\right)  }  &  =\tfrac{1}{N}\boldsymbol{\upsilon
}_{N}^{T}\mathbf{\tilde{Y}}_{N}^{T}\left(  \tfrac{1}{N}\mathbf{\tilde{Y}}%
_{N}\mathbf{\tilde{Y}}_{N}^{T}+\alpha_{M}\mathbf{I}_{M}\right)  ^{-1}\tfrac
{1}{N}\mathbf{\tilde{Y}}_{N}\mathbf{\tilde{Y}}_{N}^{T}\left(  \tfrac{1}%
{N}\mathbf{\tilde{Y}}_{N}\mathbf{\tilde{Y}}_{N}^{T}+\alpha_{M}\mathbf{I}%
_{M}\right)  ^{-1}\mathbf{\tilde{Y}}_{N}\boldsymbol{\upsilon}_{N}\nonumber\\
&  =\boldsymbol{\upsilon}_{N}^{T}\left(  \tfrac{1}{N}\mathbf{\tilde{Y}}%
_{N}^{T}\mathbf{\tilde{Y}}_{N}\right)  ^{2}\left(  \tfrac{1}{N}\mathbf{\tilde
{Y}}_{N}^{T}\mathbf{\tilde{Y}}_{N}+\alpha_{M}\mathbf{I}_{N}\right)
^{-2}\boldsymbol{\upsilon}_{N}\text{.} \label{xi6con}%
\end{align}

Before presenting the main result of this section, we provide an intermediate
result that will be required for the statement of the improved estimators.%

\begin{proposition}%
\label{propDeltaGCE}%
Under Assumptions (\textbf{As1}) to (\textbf{As3}), a generalized consistent
estimator of $\delta_{M}$, denoted by $\hat{\delta}_{M}$, is given by the
unique positive solution to the following equation:%
\[
\frac{1}{N}\operatorname*{tr}\left[  \tfrac{1}{N}\mathbf{\tilde{Y}}%
_{N}\mathbf{\tilde{Y}}_{N}^{T}\left(  \tfrac{1}{N}\mathbf{\tilde{Y}}%
_{N}\mathbf{\tilde{Y}}_{N}^{T}+\alpha_{M}\mathbf{I}_{M}\right)  ^{-1}\right]
=\delta\frac{1}{N}\operatorname*{tr}\left[  \mathbf{T}_{N}\left(
\mathbf{I}_{N}+\delta\mathbf{T}_{N}\right)  ^{-1}\right]  \text{.}%
\]%
\end{proposition}%
\begin{proof}%
The proof follows from the convergence result (\ref{propAux1a})\ in
Proposition \ref{propAux}, for $\mathbf{\Theta}_{M}=\frac{1}{N}\mathbf{I}_{M}$
and $z=-1$.%
\end{proof}%

The following theorem provides estimators of $\hat{\xi}_{M}^{\left(  4\right)
}$ and $\hat{\xi}_{M}^{\left(  6\right)  }$, which are consistent in the
double-limit regime in Assumption (\textbf{As3}).%

\begin{theorem}%
\label{theoGCE}%
(\textit{Generalized Consistent Estimators}) Under Assumptions (\textbf{As1})
to (\textbf{As3}), we have the following consistent estimators for $\hat{\xi
}_{M}^{\left(  4\right)  }$ and $\hat{\xi}_{M}^{\left(  6\right)  }$:%
\begin{align}
\hat{\xi}_{gce,M}^{\left(  4\right)  }  &  =a_{M}\hat{\xi}_{cnv,M}^{\left(
4\right)  }\text{,}\label{theoGCEa}\\
\hat{\xi}_{gce,M}^{\left(  6\right)  }  &  =a_{M}\hat{\xi}_{cnv,M}^{\left(
6\right)  }+b_{M}\text{,} \label{theoGCEb}%
\end{align}
where $\hat{\xi}_{cnv,M}^{\left(  4\right)  }$ and $\hat{\xi}_{cnv,M}^{\left(
6\right)  }$ are defined as in (\ref{xi4con}) and (\ref{xi6con}),
respectively, and%
\begin{align*}
a_{M}  &  =\frac{1}{\frac{1}{N}\operatorname*{tr}\left[  \mathbf{T}_{N}\left(
\mathbf{I}_{N}+\hat{\delta}_{M}\mathbf{T}_{N}\right)  ^{-2}\right]  }%
\text{,}\\
b_{M}  &  =-\frac{\hat{\delta}_{M}^{2}\boldsymbol{\upsilon}_{M}^{T}%
\mathbf{T}_{N}^{2}\left(  \mathbf{I}_{N}+\hat{\delta}_{M}\mathbf{T}%
_{N}\right)  ^{-2}\boldsymbol{\upsilon}_{M}}{\frac{1}{N}\operatorname*{tr}%
\left[  \mathbf{T}_{N}\left(  \mathbf{I}_{N}+\hat{\delta}_{M}\mathbf{T}%
_{N}\right)  ^{-2}\right]  }\text{,}%
\end{align*}
with $\hat{\delta}_{M}$ being given by Proposition \ref{propDeltaGCE}.%
\end{theorem}%
\begin{proof}%
See Appendix \ref{appTheoGCE}.%
\end{proof}%
%

\begin{remark}%
The asymptotic equivalents and consistent estimators of $\hat{\xi}%
_{M}^{\left(  1\right)  }$ and $\hat{\xi}_{M}^{\left(  4\right)  }$ in Theorem
and Theorem, respectively, generalize previous results on the characterization
of quadratic forms depending on the eigenvalues and eigenvectors of the sample
covariance matrix (see \cite[Proposition 1]{ML06},\cite[Chapter 4]{ML05B} and
\cite[Theorem 1]{RM09}).%
\end{remark}%
%

\begin{remark}%
We notice that Theorem \ref{theoADE}, Proposition \ref{propDeltaGCE} and
Theorem \ref{theoGCE} hold verbatim if the vectors $\boldsymbol{\upsilon}_{M}%
$, $\boldsymbol{\upsilon}_{N}$ and the matrices $\mathbf{X}_{N}$,
$\mathbf{R}_{M}$, $\mathbf{\Sigma}_{M}$, $\mathbf{T}_{N}$ have complex-valued
entries.%
\end{remark}%

\section{Numerical validations%
\label{secSim}%
}

In this section, we provide the results of some simulations illustrating the
power of the proposed analytical framework. In particular, we consider the
construction of a GMVP based on synthetic data modeling a universe of $M=50$
assets (e.g., Euro Stoxx 50) with annualized volatility (standard deviation)
between 20\% and 30\%. For simple illustration purposes, we have assumed that
the expected return is negligible compared to the asset covariance matrix, and
so it has not been estimated. We run simulations considering estimation
windows ranging from $20$ to $200$ return observations. Specifically, we
measure the accuracy of approximating the out-of-sample (realized) variance of
a GMVP by its asymptotic deterministic equivalent (ADE) given in terms of the
investment scenario parameters (cf. Section \ref{secADE}), the
\textit{conventional} (CNV) implementation based on the naive replacement of
the unknown parameters by their sample counterparts, and its
\textit{generalized consistent estimator} (GCE) derived in Section
\ref{secGCE}. Monte Carlo simulations ($10^{3}$ iterations) are run for three
different scenarios, for which the approximation error in relative terms and
in percentage is provided, i.e., $100\times\left\vert \sigma_{P}\left(
\mathbf{\hat{w}}_{\mathsf{GMVP}}\right)  -\hat{\sigma}_{P}\left(
\mathbf{\hat{w}}_{\mathsf{GMVP}}\right)  \right\vert /\sigma_{P}\left(
\mathbf{\hat{w}}_{\mathsf{GMVP}}\right)  $, where $\hat{\sigma}_{P}\left(
\mathbf{\hat{w}}_{\mathsf{GMVP}}\right)  $ denotes here any of the three
approximations. Moreover, in all cases we have considered a covariance matrix
shrinkage estimator with $\mathbf{\Sigma}_{0}=\mathbf{I}_{M}$, and parameters
$\rho$ and $\mathbf{W}_{\mathbf{\Sigma},N}$ to be calibrated for optimal
performance. In the first experiment, we consider fixed values of the
calibrating parameters given by the coefficient $\rho=0.05$ and a diagonal
matrix $\mathbf{T}=\mathbf{W}_{\mathbf{\Sigma},N}$ given by half of its
entries being equal to $t=0.75$ and the other half equal to $2-t$. Figure
\ref{figure1} shows the relative approximation error for each method.%
\begin{figure}[ptb]%
\centering
\includegraphics[
height=3.8095in,
width=5.054in
]%
{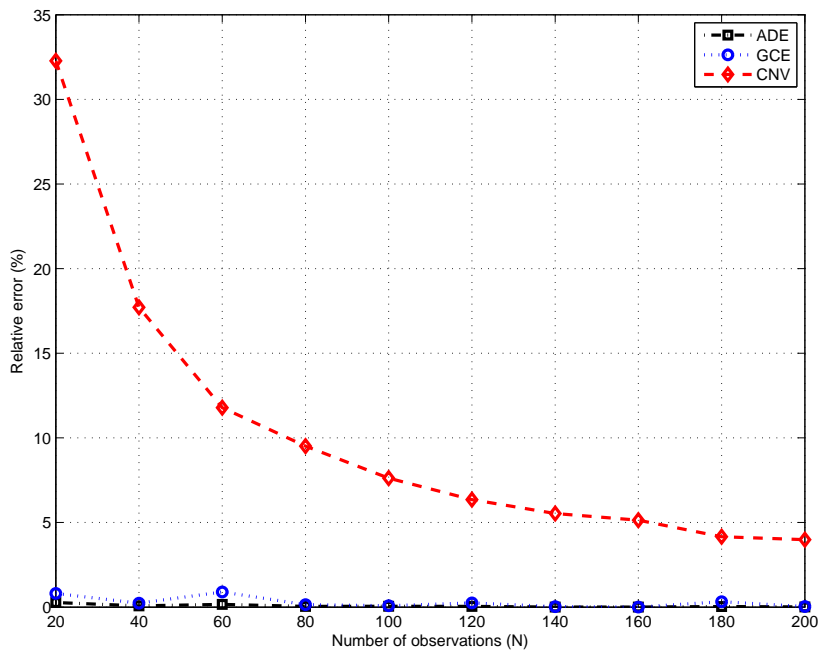}%
\caption{Approximation of realized out-of-sample variance of a GMVP for fixed
calibrating parameters}%
\label{figure1}%
\end{figure}
In the two other experiments, we consider the construction of GMVPs given by
the calibration of the optimal (for minimum variance) parameters $\rho$ or
$\mathbf{W}_{\mathbf{\Sigma},N}$, respectively, where in each case the other
parameter has been fixed to its value in the first experiment. Figures
\ref{figure2} and \ref{figure3} show the results for the calibration of $\rho$
and $\mathbf{W}_{\mathbf{\Sigma},N}$, respectively. In our simulations, we
applied a naive optimization scheme to find the optimal parameters in these
simple illustrative examples, as we do not pursue dealing with practical
optimization issues in this work, but rather focus on a representative
validation of the statistical results that we have derived; efficient
optimization algorithms based, e.g., on successive convex approximation (see
\cite{CTPOJ07}), are left as future work and are now under investigation by
the authors.%
\begin{figure}[ptb]%
\centering
\includegraphics[
height=3.8095in,
width=5.054in
]%
{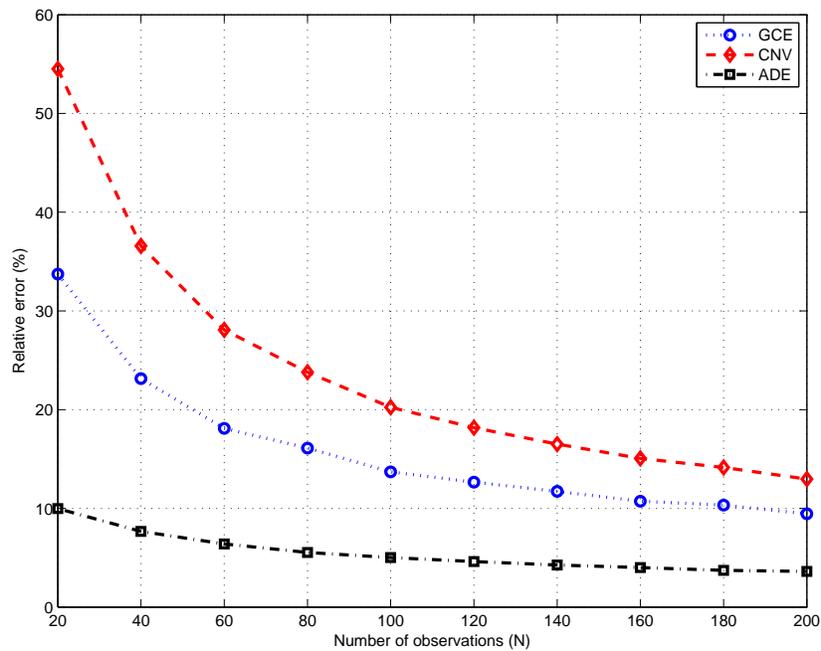}%
\caption{Approximation of realized out-of-sample variance of a GMVP for fixed
$\rho$ and optimized $\mathbf{W}_{\mathbf{\Sigma},N}$}%
\label{figure2}%
\end{figure}
\begin{figure}[ptb]%
\centering
\includegraphics[
height=3.8095in,
width=5.054in
]%
{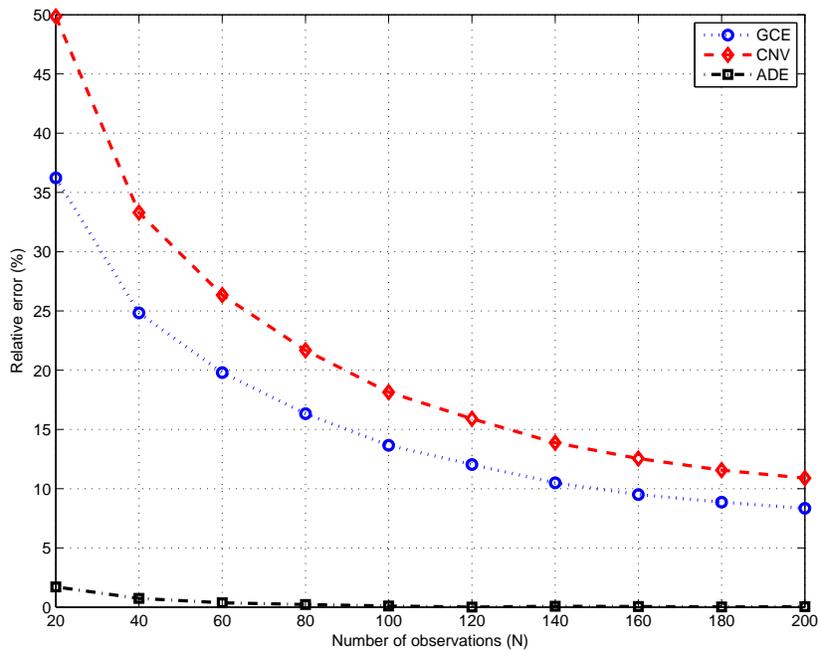}%
\caption{Approximation of realized out-of-sample variance of a GMVP for fixed
$\mathbf{W}_{\mathbf{\Sigma},N}$ and optimized $\rho$}%
\label{figure3}%
\end{figure}
From the simulation outputs, it is clear that the performance of the proposed
consistent estimators is decreased whenever calibration of the parameters has
to be performed, essentially due to the variability (fluctuations) of the
estimators. An extensive simulation campaign is outside the scope of the
section and the paper, but a reduction of this effect can be observed as
expected by increasing for instance the number of assets in the universe
(e.g., in the same illustrative line, $M=300$ for the index Euro Stoxx 300).
The use of information about the fluctuations of the estimators in order to
improve the performance of the method is currently under investigation.

\section{Conclusions%
\label{secCon}%
}

In this paper, we have provided a asymptotic framework for the analysis of the
consistency of arbitrarily large sample mean-variance portfolios that are
constructed on the basis of improved Bayesian or shrinkage estimation and
weighted sampling. To that effect, we have resorted to recent contributions on
the theory of the spectral analysis of large random matrices, based on a
double-limit regime that is defined by both the number of samples and the
number of portfolio constituents going to infinity at the same rate. In spite
of its asymptotic nature, by keeping both the return observation size and
dimension to be of the same order of magnitude our results have proved to
successfully describe the performance of sample portfolios under realistic,
finite-size situations of interest. Furthermore, based on the previous
characterization of the estimation risk, corrections of the level of risk
underestimation and return overestimation of a specific portfolio
constructions have been proposed so as to optimize the out-of-sample
performance. Our proposed calibration rules represent a sensible portfolio
choice improving on standard, usually overly optimistic Sharpe-based
investment decisions.%

\appendices

\section{Technical preliminaries%
\label{appTech}%
}

\subsection{Further definitions and auxiliary relations%
\label{appTechAux}%
}

We first recall the Sherman--Morrison--Woodbury formula, or matrix inversion
lemma, which will be used repeatedly in the sequel, i.e.,%
\begin{equation}
\left(  \mathbf{U\Xi V}+\mathbf{\Lambda}\right)  ^{-1}=\mathbf{\Lambda}%
^{-1}-\mathbf{\Lambda}^{-1}\mathbf{U}\left(  \mathbf{\Xi}^{-1}%
+\mathbf{V\Lambda}^{-1}\mathbf{U}\right)  ^{-1}\mathbf{V\Lambda}^{-1}\text{.}
\label{MIL}%
\end{equation}
In particular, the next identity for rank-augmenting matrices follows from
(\ref{MIL}):%
\begin{equation}
\left(  \mathbf{A}+\mathbf{uv}^{T}\right)  ^{-1}=\mathbf{A}^{-1}%
-\frac{\mathbf{A}^{-1}\mathbf{uv}^{T}\mathbf{A}^{-1}}{1+\mathbf{v}%
^{T}\mathbf{A}^{-1}\mathbf{u}}\text{.} \label{MILinc}%
\end{equation}
Let $\mathbf{Q}=\left(  \frac{1}{N}\mathbf{\tilde{Y}}_{N}\mathbf{\tilde{Y}%
}_{N}^{T}-z\mathbf{I}_{M}\right)  ^{-1}$, and also $\mathbf{Q}_{\left(
n\right)  }=\left(  \frac{1}{N}\mathbf{\tilde{Y}}_{\left(  n\right)
}\mathbf{\tilde{Y}}_{\left(  n\right)  }^{T}-z\mathbf{I}_{M}\right)  ^{-1}$,
where $\mathbf{\tilde{Y}}_{\left(  n\right)  }\in\mathbb{C}^{M\times N-1}$ is
defined by extracting the $n$th column from the data matrix $\mathbf{\tilde
{Y}}$. In particular, using (\ref{MILinc}), we get%
\begin{equation}
\mathbf{Q}=\mathbf{Q}_{\left(  n\right)  }-\frac{\frac{1}{N}\mathbf{Q}%
_{\left(  n\right)  }\mathbf{\tilde{y}}_{n}\mathbf{\tilde{y}}_{n}%
^{T}\mathbf{Q}_{\left(  n\right)  }}{1+\frac{1}{N}\mathbf{\tilde{y}}_{n}%
^{T}\mathbf{Q}_{\left(  n\right)  }\mathbf{\tilde{y}}_{n}}\text{.}
\label{MIL-r1u}%
\end{equation}

Using the definitions in Section \ref{secADE}, we observe that%
\begin{align}
\delta_{M}-\tilde{\delta}_{M}\gamma_{M}  &  =\alpha_{M}\frac{1}{N}%
\operatorname*{tr}\left[  \mathbf{R}_{M}\left(  \tilde{\delta}_{M}%
\mathbf{R}_{M}+\alpha_{M}\mathbf{I}_{M}\right)  ^{-2}\right] \label{ie1}\\
\tilde{\delta}_{M}-\delta_{M}\tilde{\gamma}_{M}  &  =\frac{1}{N}%
\operatorname*{tr}\left[  \mathbf{T}_{N}\left(  \mathbf{I}_{N}+\delta
_{M}\mathbf{T}_{N}\right)  ^{-2}\right]  \text{.} \label{ie2}%
\end{align}
Additionally, the following definitions will be useful for our derivations:%
\begin{align*}
\zeta_{M}  &  =\frac{1}{1-\gamma_{M}\tilde{\gamma}_{M}}\frac{1}{N}%
\operatorname*{tr}\left[  \mathbf{R}_{M}\left(  \tilde{\delta}_{M}%
\mathbf{R}_{M}+\alpha_{M}\mathbf{I}_{M}\right)  ^{-2}\right]  \text{,}\\
\tilde{\zeta}_{M}  &  =-\frac{\tilde{\gamma}_{M}}{1-\gamma_{M}\tilde{\gamma
}_{M}}\frac{1}{N}\operatorname*{tr}\left[  \mathbf{R}_{M}\left(  \tilde
{\delta}_{M}\mathbf{R}_{M}+\alpha_{M}\mathbf{I}_{M}\right)  ^{-2}\right]
\text{.}%
\end{align*}
In particular, notice that%
\begin{equation}
\tilde{\zeta}_{M}=-\tilde{\gamma}_{M}\zeta_{M}\text{.} \label{diffe}%
\end{equation}
%

\begin{lemma}%
\label{lemmaAuxRel}%
The following relations hold true:%
\begin{align}
\tilde{\delta}_{M}+\alpha_{M}\tilde{\zeta}_{M}  &  =\frac{1}{1-\gamma
_{M}\tilde{\gamma}_{M}}\frac{1}{N}\operatorname*{tr}\left[  \mathbf{T}%
_{N}\left(  \mathbf{I}_{N}+\delta_{M}\mathbf{T}_{N}\right)  ^{-2}\right]
\label{derDeltatilde}\\
\delta_{M}-\alpha_{M}\zeta_{M}  &  =\frac{\gamma_{M}}{1-\gamma_{M}%
\tilde{\gamma}_{M}}\frac{1}{N}\operatorname*{tr}\left[  \mathbf{T}_{N}\left(
\mathbf{I}_{N}+\delta_{M}\mathbf{T}_{N}\right)  ^{-2}\right]  \text{.}
\label{derDelta}%
\end{align}%
\end{lemma}%
\begin{proof}%
We first show that $\delta_{M}-\alpha_{M}\zeta_{M}=\gamma_{M}\left(
\tilde{\delta}_{M}+\alpha_{M}\tilde{\zeta}_{M}\right)  $, and then prove that
$\tilde{\delta}_{M}+\alpha_{M}\tilde{\zeta}_{M}=\left(  1-\gamma_{M}%
\tilde{\gamma}_{M}\right)  ^{-1}\left(  \tilde{\delta}_{M}-\delta_{M}%
\tilde{\gamma}_{M}\right)  $, so that the result follows finally by using
(\ref{ie2}). Let us handle the first equality. Using the definitions above, by
simple partial fraction decomposition, we get%
\[
\delta_{M}-\alpha_{M}\zeta_{M}=\gamma_{M}\left(  \tilde{\delta}_{M}-\alpha
_{M}\tilde{\gamma}_{M}\zeta_{M}\right)  \text{,}%
\]
and the first equality follows by using (\ref{diffe}). Regarding the second
equality, using the definition of $\zeta_{M}$ along with (\ref{ie1}) we notice
that%
\[
\alpha_{M}\tilde{\zeta}_{M}=\frac{1}{1-\gamma_{M}\tilde{\gamma}_{M}}\left(
\tilde{\delta}_{M}\gamma_{M}-\delta_{M}\right)  \text{,}%
\]
and the equality follows by introducing the previous expression in
$\tilde{\delta}_{M}+\alpha_{M}\tilde{\zeta}_{M}$ and finally rearranging
terms.%
\end{proof}%

\subsection{Some useful stochastic convergence results}

The following two results will be useful to prove the vanishing characteristic
of both $\hat{\xi}_{M}^{\left(  2\right)  }$ and $\hat{\xi}_{M}^{\left(
5\right)  }$.%

\begin{lemma}%
\label{lemma_BIs}
(Burkholder's inequality) Let $\left\{  \mathcal{F}_{l}\right\}  $ be a given
filtration and $\left\{  X_{l}\right\}  $ a martingale difference sequence
with respect to $\left\{  \mathcal{F}_{l}\right\}  $. Then, for any
$p\in\left(  1,\infty\right)  $, there exist constants $K_{1}$ and $K_{2}$
depending only on $p$ such that \cite[Theorem 9]{B66}%
\[
K_{1}\mathbb{E}\left[  \left(  \sum_{l=1}^{L}\left\vert X_{l}\right\vert
^{2}\right)  ^{p/2}\right]  \leq\mathbb{E}\left[  \left\vert \sum_{l=1}%
^{L}X_{l}\right\vert ^{p}\right]  \leq K_{2}\mathbb{E}\left[  \left(
\sum_{l=1}^{L}\left\vert X_{l}\right\vert ^{2}\right)  ^{p/2}\right]  \text{.}%
\]%
\end{lemma}%

The result above as well as the next were originally proved for real
variables. Extensions to the complex case are straightforward. The following
result can be shown by using the martingale convergence theorem \cite{HH80B}.
We provide a sketch of the proof, which essentially follows the exposition in
\cite[Theorem 20.10]{D94B} (see also \cite[Corollary 3]{LS06} and references therein).%

\begin{theorem}%
\label{theoremMDScon}%
Let $\left\{  \mathcal{F}_{l}\right\}  $ be a given filtration and $\left\{
X_{l}\right\}  $ a square-integrable martingale difference sequence with
respect to $\left\{  \mathcal{F}_{l}\right\}  $. If%
\[
\sup_{L\geq1}\frac{1}{L}\sum_{l=1}^{L}\mathbb{E}\left[  \left\vert
X_{l}\right\vert ^{2}\left\vert \mathcal{F}_{l-1}\right.  \right]
<\infty\text{,}%
\]
then%
\[
\frac{1}{\sqrt{L}}\sum_{l=1}^{L}X_{l}\rightarrow0\text{,}%
\]
almost surely, as $L\rightarrow\infty$.%
\end{theorem}%
%

\begin{proof}%
Define $T_{L}=L^{-1/2}\sum_{l=1}^{L}X_{l}$ so that $\left\{  T_{L}\right\}  $
is a square-integrable martingale with respect to $\left\{  \mathcal{F}%
_{l}\right\}  $. In particular, we have%
\[
\sup_{L\geq1}\mathbb{E}\left[  \left\vert T_{L}\right\vert \right]  \leq
\sup_{L\geq1}\mathbb{E}^{1/2}\left[  \left\vert T_{L}\right\vert ^{2}\right]
<\infty\text{,}%
\]
the last inequality following from Burkholder's inequality in Lemma
\ref{lemma_BIs}. Then, by the martingale convergence theorem we have that
$T_{L}$ converges almost surely as $L\rightarrow\infty$ to an integrable
random variable, and the result follows by Kronecker's lemma (see, e.g.,
\cite[pag. 31]{HH80B}).%
\end{proof}%

In the sequel, the matrix $\mathbf{\Theta}_{M}\in\mathbb{R}^{M\times M}$ will
denote an arbitrary nonrandom matrix having trace norm bounded uniformly in
$M$. Notice that $\left\Vert \mathbf{\Theta}_{M}\right\Vert _{F}\leq\left\Vert
\mathbf{\Theta}_{M}\right\Vert _{\operatorname*{tr}}$, and so the Frobenius
norm of $\mathbf{\Theta}_{M}$ is also uniformly bounded. For instance, if
$\mathbf{Z}_{M}\in\mathbb{R}^{M\times M}$ is an arbitrary nonrandom matrix
with uniformly bounded spectral (in $M$), then in the cases $\mathbf{\Theta
}_{M}=\frac{1}{M}\mathbf{Z}_{M}$ and $\mathbf{\Theta}_{M}=\boldsymbol{\upsilon
}_{M}\boldsymbol{\upsilon}_{M}^{T}$, we have $\left\Vert \frac{1}{M}%
\mathbf{Z}_{M}\right\Vert _{F}=\frac{1}{M^{1/2}}\left(  \frac{1}%
{M}\operatorname*{tr}\left[  \mathbf{Z}_{M}\mathbf{Z}_{M}^{T}\right]  \right)
^{1/2}=\mathcal{O}\left(  M^{-1/2}\right)  $ and $\left\Vert
\boldsymbol{\upsilon}_{M}\boldsymbol{\upsilon}_{M}^{T}\right\Vert
_{F}=\left\Vert \boldsymbol{\upsilon}_{M}\right\Vert ^{2}=\mathcal{O}\left(
1\right)  $, respectively. The following theorem will be instrumental in the
proof of our results. The theorem is originally stated in a more general form
for complex-valued matrices but applies verbatim for matrices with real-valued
entries \cite{RM11}.%

\begin{theorem}%
\label{theoGeneral}%
Under Assumptions (\textbf{As1}) to (\textbf{As3}), for each $z\in
\mathbb{C}-\mathbb{R}^{\mathbb{+}}$,%
\begin{align}
\operatorname*{tr}\left[  \mathbf{\Theta}_{M}\left(  \frac{1}{N}%
\mathbf{\tilde{Y}}_{N}\mathbf{\tilde{Y}}_{N}^{T}-z\mathbf{I}_{M}\right)
^{-1}\right]   &  \asymp\operatorname*{tr}\left[  \mathbf{\Theta}_{M}\left(
x_{M}\mathbf{R}-z\mathbf{I}_{M}\right)  ^{-1}\right]  \text{,}%
\label{theoGeneral_a}\\
\operatorname*{tr}\left[  \mathbf{\Theta}_{N}\left(  \frac{1}{N}%
\mathbf{\tilde{Y}}_{N}^{T}\mathbf{\tilde{Y}}_{N}-z\mathbf{I}_{N}\right)
^{-1}\right]   &  \asymp\operatorname*{tr}\left[  \mathbf{\Theta}_{N}\left(
\mathbf{I}_{N}+e_{M}\mathbf{T}\right)  ^{-1}\right]  \text{,}
\label{theoGeneral_b}%
\end{align}
where $\left\{  e_{M}=e_{M}\left(  z\right)  ,x_{M}=x_{M}\left(  z\right)
\right\}  $ is the unique solution in $\mathbb{C}^{\mathbb{+}}$ of the system
of equations:%
\[
\left\{
\begin{array}
[c]{l}%
e_{M}=\frac{1}{M}\operatorname*{tr}\left[  \mathbf{R}\left(  x_{M}%
\mathbf{T}-z\mathbf{I}_{M}\right)  ^{-1}\right] \\
x_{M}=\frac{1}{N}\operatorname*{tr}\left[  \mathbf{T}\left(  \mathbf{I}%
_{N}+e_{M}\mathbf{T}\right)  ^{-1}\right]  \text{.}%
\end{array}
\right.
\]
Moreover, given a symmetric nonnegative definite matrix $\mathbf{A}_{M}%
\in\mathbb{C}^{M\times M}$, we also have, for each $z\in\mathbb{C}%
-\mathbb{R}^{\mathbb{+}}$,%
\begin{equation}
\operatorname*{tr}\left[  \mathbf{\Theta}_{M}\left(  \mathbf{A}_{M}+\frac
{1}{N}\mathbf{\tilde{Y}}_{N}\mathbf{\tilde{Y}}_{N}^{T}-z\mathbf{I}_{M}\right)
^{-1}\right]  \asymp\operatorname*{tr}\left[  \mathbf{\Theta}_{M}\left(
\mathbf{A}_{M}+x_{M}\mathbf{R}_{M}-z\mathbf{I}_{M}\right)  ^{-1}\right]
\text{.} \label{theoGeneral_c}%
\end{equation}%
\end{theorem}%

In particular, notice that $\left\{  \tilde{\delta}_{M},\delta_{M}\right\}  $
coincides with $\left\{  e_{M}=e_{M}\left(  z\right)  ,x_{M}=x_{M}\left(
z\right)  \right\}  $ evaluated at $z=-\alpha_{M}$ (see \cite[Proposition
1]{HKLNP08}). Moreover, we remark that where $\zeta_{M}=e_{M}^{\prime}$ and
$\tilde{\zeta}_{M}=x_{M}^{\prime}$, where $e_{M}^{\prime}=e_{M}^{\prime
}\left(  z\right)  $ and $x_{M}^{\prime}=x_{M}^{\prime}\left(  z\right)  $ are
the derivatives wrt. $z$ of, respectively, $e_{M}$ and $x_{M}$, namely given
by%
\begin{align}
e_{M}^{\prime}  &  =\frac{\frac{1}{M}\operatorname*{tr}\left[  \mathbf{R}%
_{M}\left(  x_{M}\mathbf{R}_{M}-z\mathbf{I}_{M}\right)  ^{-2}\right]
}{1-\frac{1}{M}\operatorname*{tr}\left[  \mathbf{R}_{M}^{2}\left(
x_{M}\mathbf{R}_{M}-z\mathbf{I}_{M}\right)  ^{-2}\right]  \frac{1}%
{N}\operatorname*{tr}\left[  \mathbf{T}_{N}^{2}\left(  \mathbf{I}_{N}%
+e_{M}\mathbf{T}_{N}\right)  ^{-2}\right]  }\text{,}\label{eDz}\\
x_{M}^{\prime}  &  =-\frac{\frac{1}{M}\operatorname*{tr}\left[  \mathbf{R}%
_{M}\left(  x_{M}\mathbf{R}_{M}-z\mathbf{I}_{M}\right)  ^{-2}\right]  \frac
{1}{N}\operatorname*{tr}\left[  \mathbf{T}_{N}^{2}\left(  \mathbf{I}_{N}%
+e_{M}\mathbf{T}_{N}\right)  ^{-2}\right]  }{1-\frac{1}{M}\operatorname*{tr}%
\left[  \mathbf{R}_{M}^{2}\left(  x_{M}\mathbf{R}_{M}-z\mathbf{I}_{M}\right)
^{-2}\right]  \frac{1}{N}\operatorname*{tr}\left[  \mathbf{T}_{N}^{2}\left(
\mathbf{I}_{N}+e_{M}\mathbf{T}_{N}\right)  ^{-2}\right]  }\text{.} \label{xDz}%
\end{align}

Along with Theorem \ref{theoGeneral}, the following proposition will also be a
key element in proving Theorem \ref{theoADE} and Theorem \ref{theoGCE}.%

\begin{proposition}%
\label{propAux}%
Let the definitions and assumptions on the data model specified until now
hold. Then, for each $z\in\mathbb{C}-\mathbb{R}^{\mathbb{+}}$,%
\begin{align}
\operatorname*{tr}\left[  \mathbf{\Theta}_{M}\frac{1}{N}\mathbf{\tilde{Y}}%
_{N}\mathbf{\tilde{Y}}_{N}^{T}\left(  \frac{1}{N}\mathbf{\tilde{Y}}%
_{N}\mathbf{\tilde{Y}}_{N}^{T}-z\mathbf{I}_{M}\right)  ^{-1}\right]   &
\asymp x_{M}\operatorname*{tr}\left[  \mathbf{\Theta}_{M}\mathbf{R}\left(
x_{M}\mathbf{R}-z\mathbf{I}_{M}\right)  ^{-1}\right]  \text{,}%
\label{propAux1a}\\
\operatorname*{tr}\left[  \mathbf{\Theta}_{N}\frac{1}{N}\mathbf{\tilde{Y}}%
_{N}^{T}\mathbf{\tilde{Y}}_{N}\left(  \frac{1}{N}\mathbf{\tilde{Y}}_{N}%
^{T}\mathbf{\tilde{Y}}_{N}-z\mathbf{I}_{N}\right)  ^{-1}\right]   &  \asymp
e_{M}\operatorname*{tr}\left[  \mathbf{\Theta}_{N}\mathbf{T}\left(
\mathbf{I}_{N}+e_{M}\mathbf{T}\right)  ^{-1}\right]  \text{,}
\label{propAux1b}%
\end{align}
where $\left\{  e_{M}=e_{M}\left(  z\right)  ,x_{M}=x_{M}\left(  z\right)
\right\}  $ are defined as in Theorem \ref{theoGeneral}. Moreover, we also
have, for each $z\in\mathbb{C}-\mathbb{R}^{\mathbb{+}}$,%
\begin{align}
\operatorname*{tr}\left[  \mathbf{\Theta}_{M}\frac{1}{N}\mathbf{\tilde{Y}}%
_{N}\mathbf{\tilde{Y}}_{N}^{T}\left(  \frac{1}{N}\mathbf{\tilde{Y}}%
_{N}\mathbf{\tilde{Y}}_{N}^{T}-z\mathbf{I}_{M}\right)  ^{-2}\right]   &
\asymp\left(  x_{M}-zx_{M}^{\prime}\right)  \operatorname*{tr}\left[
\mathbf{\Theta}_{M}\mathbf{R}_{M}\left(  x_{M}\mathbf{R}_{M}-z\mathbf{I}%
_{M}\right)  ^{-2}\right]  \text{,}\label{propAux2a}\\
\operatorname*{tr}\left[  \mathbf{\Theta}_{N}\frac{1}{N}\mathbf{\tilde{Y}}%
_{N}^{T}\mathbf{\tilde{Y}}_{N}\left(  \frac{1}{N}\mathbf{\tilde{Y}}_{N}%
^{T}\mathbf{\tilde{Y}}_{N}-z\mathbf{I}_{N}\right)  ^{-2}\right]   &  \asymp
e_{M}^{\prime}\operatorname*{tr}\left[  \mathbf{\Theta}_{N}\mathbf{T}%
_{N}\left(  \mathbf{I}_{N}+e_{M}\mathbf{T}_{N}\right)  ^{-2}\right]  \text{,}
\label{propAux2b}%
\end{align}
where $e_{M}^{\prime}$ and $x_{M}^{\prime}$ are given by (\ref{eDz}) and
(\ref{xDz}), respectively.%
\end{proposition}%
%

\begin{proof}%
The proof of (\ref{propAux1a}) and (\ref{propAux1b}) follow the same lines of
reasoning. We show (\ref{propAux1b}). First, notice that%
\[
\operatorname*{tr}\left[  \mathbf{\Theta}_{N}\frac{1}{N}\mathbf{\tilde{Y}}%
_{N}^{T}\mathbf{\tilde{Y}}_{N}\left(  \frac{1}{N}\mathbf{\tilde{Y}}_{N}%
^{T}\mathbf{\tilde{Y}}_{N}-z\mathbf{I}_{N}\right)  ^{-1}\right]
=\operatorname*{tr}\left[  \mathbf{\Theta}_{N}\frac{1}{N}\mathbf{\tilde{Y}%
}^{T}\left(  \frac{1}{N}\mathbf{\tilde{Y}}_{N}\mathbf{\tilde{Y}}_{N}%
^{T}-z\mathbf{I}_{M}\right)  ^{-1}\mathbf{\tilde{Y}}\right]  \text{.}%
\]
Moreover, using the matrix inversion lemma in (\ref{MIL}) we get%
\[
\frac{1}{N}\mathbf{\tilde{Y}}^{T}\left(  \frac{1}{N}\mathbf{\tilde{Y}}%
_{N}\mathbf{\tilde{Y}}_{N}^{T}-z\mathbf{I}_{M}\right)  ^{-1}\mathbf{\tilde{Y}%
}=\mathbf{I}_{N}-\left(  \mathbf{I}_{N}-\frac{1}{z}\frac{1}{N}\mathbf{\tilde
{Y}}_{N}^{T}\mathbf{\tilde{Y}}_{N}\right)  ^{-1}\text{,}%
\]
and, then, write%
\begin{equation}
\operatorname*{tr}\left[  \mathbf{\Theta}_{N}\frac{1}{N}\mathbf{\tilde{Y}}%
_{N}^{T}\mathbf{\tilde{Y}}_{N}\left(  \frac{1}{N}\mathbf{\tilde{Y}}_{N}%
^{T}\mathbf{\tilde{Y}}_{N}-z\mathbf{I}_{N}\right)  ^{-1}\right]
=\operatorname*{tr}\left[  \mathbf{\Theta}_{N}\right]  +z\operatorname*{tr}%
\left[  \mathbf{\Theta}_{N}\left(  \frac{1}{N}\mathbf{\tilde{Y}}_{N}%
^{T}\mathbf{\tilde{Y}}_{N}-z\mathbf{I}_{N}\right)  ^{-1}\right]  \text{.}
\label{Lemma2MIL}%
\end{equation}
Now, Theorem \ref{theoGeneral} yields%
\[
\operatorname*{tr}\left[  \mathbf{\Theta}_{N}\left(  \frac{1}{N}%
\mathbf{\tilde{Y}}_{N}^{T}\mathbf{\tilde{Y}}_{N}-z\mathbf{I}_{N}\right)
^{-1}\right]  \asymp\operatorname*{tr}\left[  \mathbf{\Theta}_{N}\left(
\mathbf{I}_{N}+e_{M}\mathbf{T}\right)  ^{-1}\right]  \text{.}%
\]
Then, from (\ref{Lemma2MIL}), we finally have that%
\begin{align*}
\operatorname*{tr}\left[  \mathbf{\Theta}_{N}\frac{1}{N}\mathbf{\tilde{Y}}%
_{N}^{T}\mathbf{\tilde{Y}}_{N}\left(  \frac{1}{N}\mathbf{\tilde{Y}}_{N}%
^{T}\mathbf{\tilde{Y}}_{N}-z\mathbf{I}_{N}\right)  ^{-1}\right]   &
\asymp\operatorname*{tr}\left[  \mathbf{\Theta}_{N}\left(  \mathbf{I}%
_{N}-\left(  \mathbf{I}_{N}+e_{M}\mathbf{T}\right)  ^{-1}\right)  \right] \\
&  =e_{M}\operatorname*{tr}\left[  \mathbf{\Theta}_{N}\mathbf{T}\left(
\mathbf{I}_{N}+e_{M}\mathbf{T}\right)  ^{-1}\right]  \text{.}%
\end{align*}

Regarding the proof of (\ref{propAux2a}) and (\ref{propAux2b}), we first
notice that%
\begin{align*}
\operatorname*{tr}\left[  \mathbf{\Theta}_{M}\frac{1}{N}\mathbf{\tilde{Y}}%
_{N}\mathbf{\tilde{Y}}_{N}^{T}\left(  \frac{1}{N}\mathbf{\tilde{Y}}%
_{N}\mathbf{\tilde{Y}}_{N}^{T}-z\mathbf{I}_{M}\right)  ^{-2}\right]   &
=\frac{\partial}{\partial z}\left\{  \operatorname*{tr}\left[  \mathbf{\Theta
}_{M}\frac{1}{N}\mathbf{\tilde{Y}}_{N}\mathbf{\tilde{Y}}_{N}^{T}\left(
\frac{1}{N}\mathbf{\tilde{Y}}_{N}\mathbf{\tilde{Y}}_{N}^{T}-z\mathbf{I}%
_{M}\right)  ^{-1}\right]  \right\}  \text{,}\\
\operatorname*{tr}\left[  \mathbf{\Theta}\frac{1}{N}\mathbf{\tilde{Y}}_{N}%
^{T}\mathbf{\tilde{Y}}_{N}\left(  \frac{1}{N}\mathbf{\tilde{Y}}_{N}%
^{T}\mathbf{\tilde{Y}}_{N}-z\mathbf{I}_{N}\right)  ^{-2}\right]   &
=\frac{\partial}{\partial z}\left\{  \operatorname*{tr}\left[  \mathbf{\Theta
}_{M}\frac{1}{N}\mathbf{\tilde{Y}}_{N}^{T}\mathbf{\tilde{Y}}_{N}\left(
\frac{1}{N}\mathbf{\tilde{Y}}_{N}^{T}\mathbf{\tilde{Y}}_{N}-z\mathbf{I}%
_{M}\right)  ^{-1}\right]  \right\}  \text{.}%
\end{align*}
Moreover, the almost sure convergence stated in (\ref{propAux1a}) and
(\ref{propAux1b}) is uniform on $\mathbb{C}-\mathbb{R}^{\mathbb{+}}$, and
therefore the convergence of the derivatives holds by the Weierstrass
convergence theorem \cite{A78B} (see alternatively argument in \cite[Lemma
2.3]{BS04} based on Vitali's theorem about the uniform convergence of
sequences of uniformly bounded holomorphic functions towards a holomorphic
function \cite{R87B,H62B}).%
\end{proof}%

\section{Proof of Theorem \ref{theoADE}%
\label{appTheoADE}%
}

Next, we separately proof the convergence of each term in the statement of the theorem.

\subsection{The terms $\hat{\xi}_{M}^{\left(  1\right)  }$ and $\hat{\xi}%
_{M}^{\left(  3\right)  }$}

In particular, the asymptotic deterministic equivalent of $\hat{\xi}%
_{M}^{\left(  1\right)  }$ follows readily by Theorem \ref{theoGeneral}.
Regarding $\hat{\xi}_{M}^{\left(  3\right)  }$, after observing first that%
\[
\boldsymbol{\hat{\upsilon}}_{M}^{T}\mathbf{\hat{\Sigma}}_{M}^{-1}%
\boldsymbol{\hat{\upsilon}}_{M}=\boldsymbol{\upsilon}_{N}^{T}\left(  \frac
{1}{N}\mathbf{\tilde{Y}}_{N}^{T}\mathbf{\tilde{Y}}_{N}+\alpha_{M}%
\mathbf{I}_{M}\right)  ^{-1}\frac{1}{N}\mathbf{\tilde{Y}}_{N}^{T}%
\mathbf{\tilde{Y}}_{N}\mathbf{\upsilon}_{N}\text{,}%
\]
the result is obtained by applying (\ref{propAux1b}) in Proposition
\ref{propAux}.

\subsection{The term $\hat{\xi}_{M}^{\left(  2\right)  }$ and $\hat{\xi}%
_{M}^{\left(  5\right)  }$}

We recall that $\hat{\xi}_{M}^{\left(  2\right)  }=\boldsymbol{\upsilon}%
_{M}^{T}\mathbf{\hat{\Sigma}}_{M}^{-1}\boldsymbol{\hat{\upsilon}}_{M}=\frac
{1}{N}\boldsymbol{\upsilon}_{M}^{T}\mathbf{\hat{\Sigma}}_{M}^{-1}%
\mathbf{\tilde{Y}}_{N}\boldsymbol{\tilde{\upsilon}}_{N}$ and write%
\begin{equation}
\frac{1}{N}\boldsymbol{\upsilon}_{M}^{T}\mathbf{\hat{\Sigma}}_{M}%
^{-1}\mathbf{\tilde{Y}}_{N}\boldsymbol{\tilde{\upsilon}}_{N}=\frac{1}%
{N}\boldsymbol{\upsilon}_{M}^{T}\mathbf{Q}_{\left(  n\right)  }\mathbf{\tilde
{Y}}_{N}\boldsymbol{\tilde{\upsilon}}_{N}-q_{n}\frac{1}{N^{2}}%
\boldsymbol{\upsilon}_{M}^{T}\mathbf{Q}_{\left(  n\right)  }\mathbf{\tilde{y}%
}_{n}\mathbf{\tilde{y}}_{n}^{T}\mathbf{Q}_{\left(  n\right)  }\mathbf{\tilde
{Y}}_{N}\boldsymbol{\tilde{\upsilon}}_{N} \label{xi3aux1}%
\end{equation}
where we have defined $q_{n}=\left(  1+\frac{1}{N}\mathbf{y}_{n}^{T}%
\mathbf{Q}_{\left(  n\right)  }\mathbf{y}_{n}\right)  ^{-1}$. Moreover, recall
also that $\hat{\xi}_{M}^{\left(  5\right)  }=\boldsymbol{\upsilon}_{M}%
^{T}\mathbf{\hat{\Sigma}}_{M}^{-1}\mathbf{R}_{M}\mathbf{\hat{\Sigma}}_{M}%
^{-1}\boldsymbol{\hat{\upsilon}}_{M}=\frac{1}{N}\boldsymbol{\upsilon}_{M}%
^{T}\mathbf{\hat{\Sigma}}_{M}^{-1}\mathbf{R}_{M}\mathbf{\hat{\Sigma}}_{M}%
^{-1}\mathbf{\tilde{Y}}_{N}\boldsymbol{\tilde{\upsilon}}_{N}$, and consider
first the following notations:%
\[
\chi_{j,n}=\frac{1}{N}\mathbf{\tilde{y}}_{n}^{T}\mathbf{Q}_{\left(  n\right)
}\mathbf{Z}_{j\left(  n\right)  }\mathbf{\tilde{y}}_{n}-\frac{1}%
{N}\operatorname*{tr}\left[  \mathbf{Q}_{\left(  n\right)  }\mathbf{Z}%
_{j\left(  n\right)  }\right]  \text{,}%
\]
for $j=1,2$, with $\mathbf{Z}_{j\left(  n\right)  }$ being arbitrary $M\times
M$ dimensional matrices, possibly random but not depending on the
$\mathbf{x}_{n}$, and such that $\sup_{M\geq1}\left\Vert \mathbf{Z}_{j\left(
n\right)  }\right\Vert <+\infty$; in particular, $\mathbf{Z}_{1\left(
n\right)  }=\mathbf{I}_{M}$ and $\mathbf{Z}_{2\left(  n\right)  }%
=\mathbf{Q}_{\left(  n\right)  }\mathbf{R}_{M}$. Then, observe that we can
write $\hat{\xi}_{M}^{\left(  2\right)  }$ and $\hat{\xi}_{M}^{\left(
5\right)  }$ as, respectively,%
\begin{align}
\frac{1}{N}\boldsymbol{\upsilon}_{M}^{T}\mathbf{Q}_{\left(  n\right)
}\mathbf{\tilde{Y}}_{N}\boldsymbol{\tilde{\upsilon}}_{N}  &  =\frac{1}{N}%
\sum_{n=1}^{N}\mathbf{z}_{1\left(  n\right)  }^{T}\mathbf{x}_{n}+\frac{1}%
{N}\sum_{n=1}^{N}\mathbf{z}_{2\left(  n\right)  }^{T}\mathbf{x}_{n}\nonumber\\
&  +\frac{1}{N}\sum_{n=1}^{N}\chi_{1,n}t_{n}q_{n}\left[  \boldsymbol{\tilde
{\upsilon}}_{N}\right]  _{n}\mathbf{\tilde{\upsilon}}_{M}^{T}\mathbf{Q}%
_{\left(  n\right)  }\mathbf{Z}_{1\left(  n\right)  }\mathbf{\tilde{y}}%
_{n}\text{,} \label{vanish1}%
\end{align}
where the first term on the RHS follows from $\frac{1}{N}\boldsymbol{\upsilon
}_{M}^{T}\mathbf{Q}_{\left(  n\right)  }\mathbf{\tilde{Y}}_{N}%
\boldsymbol{\tilde{\upsilon}}_{N}$, and, similarly,%
\begin{align}
\frac{1}{N}\boldsymbol{\upsilon}_{M}^{T}\mathbf{\hat{\Sigma}}_{M}%
^{-1}\mathbf{R}_{M}\mathbf{\hat{\Sigma}}_{M}^{-1}\mathbf{\tilde{Y}}%
_{N}\boldsymbol{\tilde{\upsilon}}_{N}  &  =\frac{1}{N}\sum_{n=1}^{N}%
\mathbf{z}_{3\left(  n\right)  }^{T}\mathbf{x}_{n}-\frac{1}{N}\sum_{n=1}%
^{N}\mathbf{z}_{4\left(  n\right)  }^{T}\mathbf{x}_{n}-\frac{1}{N}\sum
_{n=1}^{N}\mathbf{z}_{5\left(  n\right)  }^{T}\mathbf{x}_{n}\nonumber\\
&  -\frac{1}{N}\sum_{n=1}^{N}\chi_{1,n}t_{n}q_{n}\left[  \boldsymbol{\tilde
{\upsilon}}_{N}\right]  _{n}\boldsymbol{\upsilon}_{M}^{T}\mathbf{Q}_{\left(
n\right)  }\mathbf{Z}_{2\left(  n\right)  }\mathbf{\tilde{y}}_{n}-\frac{1}%
{N}\sum_{n=1}^{N}\chi_{2,n}t_{n}q_{n}\left[  \boldsymbol{\tilde{\upsilon}}%
_{N}\right]  _{n}\boldsymbol{\upsilon}_{M}^{T}\mathbf{Q}_{\left(  n\right)
}\mathbf{Z}_{1\left(  n\right)  }\mathbf{\tilde{y}}_{n}\nonumber\\
&  +\frac{1}{N}\sum_{n=1}^{N}t_{n}q_{n}^{2}\left[  \boldsymbol{\tilde
{\upsilon}}_{N}\right]  _{n}\frac{1}{N}\boldsymbol{\upsilon}_{M}^{T}%
\mathbf{Q}_{\left(  n\right)  }\mathbf{Z}_{1\left(  n\right)  }\mathbf{\tilde
{y}}_{n}\frac{1}{N}\mathbf{\tilde{y}}_{n}^{T}\mathbf{Q}_{\left(  n\right)
}\mathbf{Z}_{2\left(  n\right)  }\mathbf{\tilde{y}}_{n}\frac{1}{N}%
\mathbf{\tilde{y}}_{n}^{T}\mathbf{Q}_{\left(  n\right)  }\mathbf{Z}_{1\left(
n\right)  }\mathbf{\tilde{y}}_{n}\text{,} \label{vanish2}%
\end{align}
where the following definitions apply:%
\begin{align*}
\mathbf{z}_{1\left(  n\right)  }  &  =t_{n}\left[  \boldsymbol{\tilde
{\upsilon}}_{N}\right]  _{n}\mathbf{R}_{M}^{1/2}\mathbf{Z}_{1\left(  n\right)
}\mathbf{Q}_{\left(  n\right)  }\boldsymbol{\upsilon}_{M}\text{,}\\
\mathbf{z}_{2\left(  n\right)  }  &  =t_{n}q_{n}\left[  \boldsymbol{\tilde
{\upsilon}}_{N}\right]  _{n}\frac{1}{N}\operatorname*{tr}\left[
\mathbf{Q}_{\left(  n\right)  }\mathbf{Z}_{1\left(  n\right)  }\right]
\mathbf{R}_{M}^{1/2}\mathbf{Z}_{1\left(  n\right)  }\mathbf{Q}_{\left(
n\right)  }\boldsymbol{\upsilon}_{M}\text{,}\\
\mathbf{z}_{3\left(  n\right)  }  &  =t_{n}\left[  \boldsymbol{\tilde
{\upsilon}}_{N}\right]  _{n}\mathbf{R}_{M}^{1/2}\mathbf{Z}_{2\left(  n\right)
}\mathbf{Q}_{\left(  n\right)  }\mathbf{\tilde{\upsilon}}_{M}\text{,}\\
\mathbf{z}_{4\left(  n\right)  }  &  =t_{n}q_{n}\left[  \boldsymbol{\tilde
{\upsilon}}_{N}\right]  _{n}\frac{1}{N}\operatorname*{tr}\left[
\mathbf{Q}_{\left(  n\right)  }\right]  \mathbf{R}_{M}^{1/2}\mathbf{Z}%
_{2\left(  n\right)  }\mathbf{Q}_{\left(  n\right)  }\boldsymbol{\upsilon}%
_{M}\text{,}\\
\mathbf{z}_{5\left(  n\right)  }  &  =t_{n}q_{n}\left[  \boldsymbol{\tilde
{\upsilon}}_{N}\right]  _{n}\frac{1}{N}\operatorname*{tr}\left[
\mathbf{Q}_{\left(  n\right)  }\mathbf{Z}_{2\left(  n\right)  }\right]
\mathbf{R}_{M}^{1/2}\mathbf{Q}_{\left(  n\right)  }\boldsymbol{\upsilon}%
_{M}\text{.}%
\end{align*}

We now prove that the terms of the form $\frac{1}{N}\sum_{n=1}^{N}%
\mathbf{z}_{k\left(  n\right)  }^{T}\mathbf{x}_{n}$ vanish almost surely. To
see this, we further define the following two $L=MN$ dimensional vectors,
namely, $\mathbf{x}=\left[
\begin{array}
[c]{ccc}%
\mathbf{x}_{1}^{T} & \cdots & \mathbf{x}_{N}^{T}%
\end{array}
\right]  $, and $\mathbf{z}_{k}=\left[
\begin{array}
[c]{ccc}%
\mathbf{z}_{k\left(  1\right)  }^{T} & \cdots & \mathbf{z}_{k\left(  N\right)
}^{T}%
\end{array}
\right]  $, $k=1,2,3,4,5$. Then, we notice that%
\begin{equation}
\frac{1}{N}\sum_{n=1}^{N}\mathbf{z}_{k\left(  n\right)  }^{T}\mathbf{x}%
_{n}=\frac{1}{N}\mathbf{z}_{k}^{T}\mathbf{x}\equiv\frac{1}{\sqrt{L}}\sum
_{l=1}^{L}\eta_{l}\text{,} \label{zk}%
\end{equation}
where we have defined $\eta_{k,l}=Z_{k,l}^{T}X_{l}$, with $Z_{k,l}$ and
$X_{l}$ being the $l$th entries of $\mathbf{z}_{k}$ and $\mathbf{x}$,
respectively. In particular, if $\mathcal{G}_{l}$ is the $\sigma$-field
generated by the random variables $\left\{  X_{l}\right\}  $, then notice that
$\left\{  \eta_{k,l}\right\}  $ forms a martingale difference sequence with
respect to the filtration $\left\{  \mathcal{G}_{l}\right\}  $. Indeed,
$\mathbb{E}\left[  \left.  Z_{k,l}^{T}X_{l}\right\vert \mathcal{G}%
_{l-1}\right]  =0$. Then, we notice that $\mathbb{E}\left[  \left.  \left\vert
\eta_{k,l}\right\vert ^{2}\right\vert \mathcal{G}_{l-1}\right]  =\mathbb{E}%
\left[  \left.  \left\vert Z_{k,l}\right\vert ^{2}\right\vert \mathcal{G}%
_{l-1}\right]  $ is bounded by assumption and so by Theorem
\ref{theoremMDScon} we have that (\ref{zk}) vanishes almost surely.

We now handle the last term on the RHS of equation (\ref{vanish1}) together
with the last three terms on the RHS of equation (\ref{vanish2}). From the
developments in \cite{RM11} it follows that, for a sufficiently large $p$,%
\begin{equation}
\mathbb{E}\left[  \left\vert t_{n}q_{n}\left[  \boldsymbol{\tilde{\upsilon}%
}_{N}\right]  _{n}\boldsymbol{\upsilon}_{M}^{T}\mathbf{Q}_{\left(  n\right)
}\mathbf{Z}_{j\left(  n\right)  }\mathbf{\tilde{y}}_{n}\right\vert
^{2p}\right]  \leq\frac{K_{p}\left\vert z\right\vert ^{2p}}{\left\vert
\operatorname{Im}\left\{  z\right\}  \right\vert ^{4p}}\text{,}
\label{inequalityOne}%
\end{equation}
where $K_{p}$ is a constant depending on $p$ but not on $M,N$ which may take
different values at each appearance, and%
\begin{equation}
\mathbb{E}\left[  \left\vert \chi_{j,n}\right\vert ^{2p}\right]  \leq\frac
{1}{N^{p}}\frac{K_{p}}{\left\vert \operatorname{Im}\left\{  z\right\}
\right\vert ^{2p}}\text{.} \label{inequalityTwo}%
\end{equation}
Then, using (\ref{inequalityOne}) and (\ref{inequalityTwo}), and applying
first Minkowski's and then the Cauchy-Schwarz inequalities, we get ($i,j=1,2$)%
\begin{multline*}
\mathbb{E}\left[  \left\vert \frac{1}{N}\sum_{n=1}^{N}\chi_{i,n}t_{n}%
q_{n}\left[  \boldsymbol{\tilde{\upsilon}}_{N}\right]  _{n}%
\boldsymbol{\upsilon}_{M}^{T}\mathbf{Q}_{\left(  n\right)  }\mathbf{Z}%
_{j\left(  n\right)  }\mathbf{\tilde{y}}_{n}\right\vert ^{p}\right]  \leq
\frac{1}{N^{p}}\left(  \sum_{n=1}^{N}\left(  \mathbb{E}\left[  \left\vert
\chi_{i,n}t_{n}q_{n}\left[  \boldsymbol{\tilde{\upsilon}}_{N}\right]
_{n}\boldsymbol{\upsilon}_{M}^{T}\mathbf{Q}_{\left(  n\right)  }%
\mathbf{Z}_{j\left(  n\right)  }\mathbf{\tilde{y}}_{n}\right\vert ^{p}\right]
\right)  ^{1/p}\right)  ^{p}\\
\leq\frac{1}{N^{p}}\left(  \sum_{n=1}^{N}\mathbb{E}^{1/2p}\left[  \left\vert
t_{n}q_{n}\left[  \boldsymbol{\tilde{\upsilon}}_{N}\right]  _{n}%
\boldsymbol{\upsilon}_{M}^{T}\mathbf{Q}_{\left(  n\right)  }\mathbf{Z}%
_{j\left(  n\right)  }\mathbf{\tilde{y}}_{n}\right\vert ^{2p}\right]
\mathbb{E}^{1/2p}\left[  \left\vert \chi_{i,n}\right\vert ^{2p}\right]
\right)  ^{p}\leq\frac{1}{N^{p/2}}\frac{K_{p}\left\vert z\right\vert ^{p}%
}{\left\vert \operatorname{Im}\left\{  z\right\}  \right\vert ^{3p}}\text{.}%
\end{multline*}
Furthermore, let $\mathcal{X}_{n}=\frac{1}{N}\sum_{n=1}^{N}t_{n}q_{n}%
^{2}\left[  \boldsymbol{\tilde{\upsilon}}_{N}\right]  _{n}\frac{1}%
{N}\boldsymbol{\upsilon}_{M}^{T}\mathbf{Q}_{\left(  n\right)  }\mathbf{Z}%
_{1\left(  n\right)  }\mathbf{\tilde{y}}_{n}\frac{1}{N}\mathbf{\tilde{y}}%
_{n}^{T}\mathbf{Q}_{\left(  n\right)  }\mathbf{Z}_{2\left(  n\right)
}\mathbf{\tilde{y}}_{n}\frac{1}{N}\mathbf{\tilde{y}}_{n}^{T}\mathbf{Q}%
_{\left(  n\right)  }\mathbf{Z}_{1\left(  n\right)  }\mathbf{\tilde{y}}_{n}$
and, by using Jensen's inequality along with (\ref{inequalityOne}), observe
that%
\[
\mathbb{E}\left[  \left\vert \frac{1}{N}\sum_{n=1}^{N}\mathcal{X}%
_{n}\right\vert ^{p}\right]  \leq\frac{1}{N}\sum_{n=1}^{N}\mathbb{E}\left[
\left\vert \mathcal{X}_{n}\right\vert ^{p}\right]  \leq\max_{1\leq n\leq
N}\mathbb{E}\left[  \left\vert \mathcal{X}_{n}\right\vert ^{p}\right]
\leq\frac{1}{N^{1+\epsilon}}\frac{K_{p}\left\vert z\right\vert ^{q}%
}{\left\vert \operatorname{Im}\left\{  z\right\}  \right\vert ^{r}%
}\text{,\quad}\epsilon>0\text{,}%
\]
for sufficiently large $p$ and well-chosen $q$ and $r$, on $p$ but not on
$M,N$. Then, the almost sure convergence to zero of the four terms for each
$z\in\mathbb{C}^{\mathbb{+}}$ follows then by Borel-Cantelli's lemma. Finally,
convergence of the real nonnegative axis follows by an argument based on
Montel's normal family theorem (see, e.g., Section 4 in \cite{RM11}).

\subsection{The term $\hat{\xi}_{M}^{\left(  4\right)  }$ and $\hat{\xi}%
_{M}^{\left(  6\right)  }$}

Observe that%
\begin{align*}
\boldsymbol{\upsilon}_{M}^{T}\mathbf{\hat{\Sigma}}_{M}^{-1}\mathbf{R}%
_{M}\mathbf{\hat{\Sigma}}_{M}^{-1}\boldsymbol{\upsilon}_{M}  &
=\boldsymbol{\upsilon}_{M}^{T}\mathbf{R}_{M}^{-1/2}\left(  \mathbf{X}%
_{N}\mathbf{T}_{N}\mathbf{X}_{N}^{T}+\alpha_{M}\mathbf{R}_{M}^{-1}\right)
^{-2}\mathbf{R}_{M}^{-1/2}\boldsymbol{\upsilon}_{M}\\
&  =\left.  \frac{\partial}{\partial z}\left\{  \boldsymbol{\upsilon}_{M}%
^{T}\mathbf{R}_{M}^{-1/2}\left(  \alpha_{M}\mathbf{R}_{M}^{-1}+\mathbf{X}%
_{N}\mathbf{T}_{N}\mathbf{X}_{N}^{T}-z\mathbf{I}_{M}\right)  ^{-1}%
\mathbf{R}_{M}^{-1/2}\boldsymbol{\upsilon}_{M}\right\}  \right\vert
_{z=0}\text{.}%
\end{align*}
Furthermore, using (\ref{theoGeneral_c}) in Theorem \ref{theoGeneral} we get%
\[
\boldsymbol{\upsilon}_{M}^{T}\mathbf{R}_{M}^{-1/2}\left(  \alpha_{M}%
\mathbf{R}_{M}^{-1}+\mathbf{X}_{N}\mathbf{T}_{N}\mathbf{X}_{N}^{T}%
-z\mathbf{I}_{M}\right)  ^{-1}\mathbf{R}_{M}^{-1/2}\boldsymbol{\upsilon}%
_{M}\asymp\boldsymbol{\upsilon}_{M}^{T}\mathbf{R}_{M}^{-1/2}\left(  \alpha
_{M}\mathbf{R}_{M}^{-1}+\left(  x_{M}^{\left(  4\right)  }-z\right)
\mathbf{I}_{M}\right)  ^{-1}\mathbf{R}_{M}^{-1/2}\boldsymbol{\upsilon}%
_{M}\text{,}%
\]
where, for each $z$ outside the real positive axis, $\left\{  e_{M}^{\left(
4\right)  }\left(  z\right)  ,x_{M}^{\left(  4\right)  }\left(  z\right)
\right\}  $ is the unique solution to the system:%
\begin{align*}
e_{M}^{\left(  4\right)  }\left(  z\right)   &  =\frac{1}{N}\operatorname*{tr}%
\left[  \left(  \alpha_{M}\mathbf{R}_{M}^{-1}+\left(  x_{M}^{\left(  4\right)
}\left(  z\right)  -z\right)  \mathbf{I}_{M}\right)  ^{-1}\right] \\
x_{M}^{\left(  4\right)  }\left(  z\right)   &  =\frac{1}{N}\operatorname*{tr}%
\left[  \mathbf{T}_{N}\left(  \mathbf{I}_{N}+e_{M}^{\left(  4\right)  }\left(
z\right)  \mathbf{T}_{N}\right)  ^{-1}\right]  \text{.}%
\end{align*}
Then, using%
\[
x_{M}^{\left(  4\right)  \prime}\left(  0\right)  =-\left(  1-x_{M}^{\left(
4\right)  \prime}\left(  0\right)  \right)  \frac{1}{N}\operatorname*{tr}%
\left[  \mathbf{R}_{M}^{2}\left(  x_{M}^{\left(  4\right)  }\left(  0\right)
\mathbf{R}_{M}+\alpha_{M}\mathbf{I}_{M}\right)  ^{-2}\right]  \frac{1}%
{N}\operatorname*{tr}\left[  \mathbf{T}_{N}^{2}\left(  \mathbf{I}_{N}%
+e_{M}^{\left(  4\right)  }\left(  0\right)  \mathbf{T}_{N}\right)
^{-2}\right]  \text{,}%
\]
we finally get%
\begin{multline*}
\left.  \frac{\partial}{\partial z}\left\{  \boldsymbol{\upsilon}_{M}%
^{T}\mathbf{R}_{M}^{-1/2}\left(  \alpha_{M}\mathbf{R}_{M}^{-1}+\left(
x_{M}^{\left(  4\right)  }-z\right)  \mathbf{I}_{M}\right)  ^{-1}%
\mathbf{R}_{M}^{-1/2}\boldsymbol{\upsilon}_{M}\right\}  \right\vert _{z=0}\\
=\left(  1-x_{M}^{\left(  4\right)  \prime}\left(  0\right)  \right)
\boldsymbol{\upsilon}_{M}^{T}\mathbf{R}_{M}^{-1/2}\left(  x_{M}^{\left(
4\right)  }\left(  0\right)  \mathbf{R}_{M}+\alpha_{M}\mathbf{I}_{M}\right)
^{-2}\mathbf{R}_{M}^{-1/2}\boldsymbol{\upsilon}_{M}\text{,}%
\end{multline*}
where%
\begin{align*}
1-x_{M}^{\left(  4\right)  \prime}\left(  0\right)   &  =1+e_{M}^{\left(
4\right)  \prime}\left(  0\right)  \frac{1}{N}\operatorname*{tr}\left[
\mathbf{T}_{N}^{2}\left(  \mathbf{I}_{N}+e_{M}^{\left(  4\right)  }\left(
0\right)  \mathbf{T}_{N}\right)  ^{-2}\right] \\
&  =1+\left(  1-x_{M}^{\left(  4\right)  \prime}\left(  0\right)  \right)
\frac{1}{N}\operatorname*{tr}\left[  \mathbf{R}_{M}^{2}\left(  x_{M}^{\left(
4\right)  }\left(  0\right)  \mathbf{R}_{M}+\alpha_{M}\mathbf{I}_{M}\right)
^{-2}\right]  \frac{1}{N}\operatorname*{tr}\left[  \mathbf{T}_{N}^{2}\left(
\mathbf{I}_{N}+e_{M}^{\left(  4\right)  }\left(  0\right)  \mathbf{T}%
_{N}\right)  ^{-2}\right] \\
&  =\frac{1}{1-\frac{1}{N}\operatorname*{tr}\left[  \mathbf{T}_{N}^{2}\left(
\mathbf{I}_{N}+e_{M}^{\left(  4\right)  }\left(  0\right)  \mathbf{T}%
_{N}\right)  ^{-2}\right]  \frac{1}{N}\operatorname*{tr}\left[  \mathbf{R}%
_{M}^{2}\left(  x_{M}^{\left(  4\right)  }\left(  0\right)  \mathbf{R}%
_{M}+\alpha_{M}\mathbf{I}_{M}\right)  ^{-2}\right]  }\text{.}%
\end{align*}

Let us now deal with $\hat{\xi}_{M}^{\left(  6\right)  }$. We recall that%
\[
\hat{\xi}_{M}^{\left(  6\right)  }=\boldsymbol{\hat{\upsilon}}_{M}%
^{T}\mathbf{\hat{\Sigma}}_{M}^{-1}\mathbf{R}_{M}\mathbf{\hat{\Sigma}}_{M}%
^{-1}\boldsymbol{\hat{\upsilon}}_{M}=\frac{1}{N}\boldsymbol{\upsilon}_{N}%
^{T}\mathbf{T}_{N}^{1/2}\mathbf{X}^{T}\left(  \frac{1}{N}\mathbf{X}%
_{N}\mathbf{T}_{N}\mathbf{X}_{N}^{T}+\alpha_{M}\mathbf{R}_{M}^{-1}\right)
^{-2}\mathbf{X}_{N}\mathbf{T}_{N}^{1/2}\boldsymbol{\upsilon}_{N}\text{.}%
\]
Let $\mathbf{A}_{M}=\mathbf{A}_{M}\left(  t\right)  =\alpha_{M}\mathbf{R}%
_{M}^{-1}+t\mathbf{I}_{M}$, with $t>0$ being a real positive scalar, and
observe that%
\begin{multline*}
\boldsymbol{\upsilon}_{N}^{T}\mathbf{T}_{N}^{1/2}\mathbf{X}^{T}\left(
\frac{1}{N}\mathbf{X}_{N}\mathbf{T}_{N}\mathbf{X}_{N}^{T}+\alpha_{M}%
\mathbf{R}_{M}^{-1}\right)  ^{-2}\mathbf{X}_{N}\mathbf{T}_{N}^{1/2}%
\boldsymbol{\upsilon}_{N}\\
=-\left.  \frac{\partial}{\partial t}\left\{  \boldsymbol{\upsilon}_{N}%
^{T}\mathbf{T}_{N}^{1/2}\mathbf{X}^{T}\left(  \frac{1}{N}\mathbf{X}%
_{N}\mathbf{T}_{N}\mathbf{X}_{N}^{T}+\mathbf{A}_{M}\left(  t\right)  \right)
^{-1}\mathbf{X}_{N}\mathbf{T}_{N}^{1/2}\boldsymbol{\upsilon}_{N}\right\}
\right\vert _{t=0}\text{.}%
\end{multline*}
Furthermore, using the matrix inversion lemma in (\ref{MIL}), we write%
\[
\mathbf{T}_{N}^{1/2}\mathbf{X}^{T}\left(  \frac{1}{N}\mathbf{X}_{N}%
\mathbf{T}_{N}\mathbf{X}_{N}^{T}+\mathbf{A}_{M}\left(  t\right)  \right)
^{-1}\mathbf{X}_{N}\mathbf{T}_{N}^{1/2}=\mathbf{I}_{N}-\left(  \mathbf{I}%
_{N}+\mathbf{T}_{N}^{1/2}\mathbf{X}^{T}\left(  \alpha_{M}\mathbf{R}_{M}%
^{-1}+t\mathbf{I}_{M}\right)  ^{-1}\mathbf{X}_{N}\mathbf{T}_{N}^{1/2}\right)
^{-1}\text{,}%
\]
and so we have%
\begin{multline*}
\boldsymbol{\upsilon}_{N}^{T}\mathbf{T}_{N}^{1/2}\mathbf{X}^{T}\left(
\frac{1}{N}\mathbf{X}_{N}\mathbf{T}_{N}\mathbf{X}_{N}^{T}+\alpha
_{M}\mathbf{\tilde{\Sigma}}\right)  ^{-2}\mathbf{X}_{N}\mathbf{T}_{N}%
^{1/2}\boldsymbol{\upsilon}_{N}\\
=\left.  \frac{\partial}{\partial t}\left\{  \boldsymbol{\upsilon}_{N}%
^{T}\left(  \mathbf{I}_{N}+\mathbf{T}_{N}^{1/2}\mathbf{X}^{T}\left(
\alpha_{M}\mathbf{R}_{M}^{-1}+t\mathbf{I}_{M}\right)  ^{-1}\mathbf{X}%
_{N}\mathbf{T}_{N}^{1/2}\right)  ^{-1}\boldsymbol{\upsilon}_{N}\right\}
\right\vert _{t=0}\text{.}%
\end{multline*}
Now, using Theorem \ref{theoGeneral} we get%
\[
\boldsymbol{\upsilon}_{N}^{T}\left(  \frac{1}{N}\mathbf{T}_{N}^{1/2}%
\mathbf{X}^{T}\left(  \alpha_{M}\mathbf{R}_{M}^{-1}+t\mathbf{I}_{M}\right)
^{-1}\mathbf{X}_{N}\mathbf{T}_{N}^{1/2}+\mathbf{I}_{N}\right)  ^{-1}%
\boldsymbol{\upsilon}_{N}\asymp\boldsymbol{\upsilon}_{N}^{T}\left(
x_{M}^{\left(  6\right)  }\left(  -1\right)  \mathbf{T}_{N}+\mathbf{I}%
_{N}\right)  ^{-1}\boldsymbol{\upsilon}_{N}\text{,}%
\]
where $\left\{  x_{M}^{\left(  6\right)  }\left(  -1\right)  =x_{M}^{\left(
6\right)  }\left(  t\right)  ,e_{M}^{\left(  6\right)  }\left(  -1\right)
=e_{M}^{\left(  6\right)  }\left(  t\right)  \right\}  $ is the solution to
the following system of equations:%
\begin{align*}
e_{M}^{\left(  6\right)  }\left(  t\right)   &  =\frac{1}{N}\operatorname*{tr}%
\left[  \mathbf{T}_{N}\left(  x_{M}^{\left(  6\right)  }\left(  t\right)
\mathbf{T}_{N}+\mathbf{I}_{M}\right)  ^{-1}\right] \\
x_{M}^{\left(  6\right)  }\left(  t\right)   &  =\frac{1}{N}\operatorname*{tr}%
\left[  \left(  \alpha_{M}\mathbf{R}_{M}^{-1}+\left(  t+e_{M}^{\left(
6\right)  }\left(  t\right)  \right)  \mathbf{I}_{M}\right)  ^{-1}\right]
\text{.}%
\end{align*}
Finally, notice that%
\[
\left.  \frac{\partial}{\partial t}\left\{  \boldsymbol{\upsilon}_{N}%
^{T}\left(  x_{M}^{\left(  6\right)  }\left(  t\right)  \mathbf{T}%
_{N}+\mathbf{I}_{N}\right)  ^{-1}\boldsymbol{\upsilon}_{N}\right\}
\right\vert _{t=0}=-x_{M}^{\left(  6\right)  \prime}\left(  0\right)
\boldsymbol{\upsilon}_{N}^{T}\left(  x_{M}^{\left(  6\right)  }\left(
0\right)  \mathbf{T}_{N}+\mathbf{I}_{N}\right)  ^{-2}\boldsymbol{\upsilon}%
_{N}\text{,}%
\]
where $x_{M}^{\left(  6\right)  }\left(  0\right)  =\frac{1}{N}%
\operatorname*{tr}\left[  \left(  \alpha_{M}\mathbf{R}_{M}^{-1}+e_{M}^{\left(
6\right)  }\left(  0\right)  \mathbf{I}_{M}\right)  ^{-1}\right]  $, and%
\begin{align*}
-x_{M}^{\left(  6\right)  \prime}\left(  0\right)   &  =\left(  1+e_{M}%
^{\left(  6\right)  \prime}\left(  0\right)  \right)  \frac{1}{N}%
\operatorname*{tr}\left[  \mathbf{\Sigma}^{2}\left(  e_{M}^{\left(  6\right)
}\left(  0\right)  \mathbf{R}_{M}+\alpha_{M}\mathbf{I}_{M}\right)
^{-2}\right] \\
&  =\left(  1-x_{M}^{\left(  6\right)  \prime}\left(  0\right)  \frac{1}%
{N}\operatorname*{tr}\left[  \mathbf{T}_{N}\left(  x_{M}^{\left(  6\right)
}\left(  0\right)  \mathbf{T}_{N}+\mathbf{I}_{M}\right)  ^{-2}\right]
\right)  \frac{1}{N}\operatorname*{tr}\left[  \mathbf{R}_{M}^{2}\left(
e_{M}^{\left(  6\right)  }\left(  0\right)  \mathbf{R}_{M}+\alpha
_{M}\mathbf{I}_{M}\right)  ^{-2}\right] \\
&  =\frac{\frac{1}{N}\operatorname*{tr}\left[  \mathbf{R}_{M}^{2}\left(
e_{M}^{\left(  6\right)  }\left(  0\right)  \mathbf{R}_{M}+\alpha
_{M}\mathbf{I}_{M}\right)  ^{-2}\right]  }{1-\frac{1}{N}\operatorname*{tr}%
\left[  \mathbf{R}_{M}^{2}\left(  e_{M}^{\left(  6\right)  }\left(  0\right)
\mathbf{R}_{M}+\alpha_{M}\mathbf{I}_{M}\right)  ^{-2}\right]  \frac{1}%
{N}\operatorname*{tr}\left[  \mathbf{T}_{N}\left(  x_{M}^{\left(  6\right)
}\left(  0\right)  \mathbf{T}_{N}+\mathbf{I}_{M}\right)  ^{-2}\right]
}\text{.}%
\end{align*}

\section{Proof of Theorem \ref{theoGCE}%
\label{appTheoGCE}%
}

We first show (\ref{theoGCEa}), i.e.,%
\begin{multline*}
\frac{1}{\frac{1}{N}\operatorname*{tr}\left[  \mathbf{T}_{N}\left(
\mathbf{I}_{N}+\hat{\delta}_{M}\mathbf{T}_{N}\right)  ^{-2}\right]
}\boldsymbol{\upsilon}_{M}^{T}\frac{1}{N}\mathbf{\tilde{Y}}_{N}\mathbf{\tilde
{Y}}_{N}^{T}\left(  \frac{1}{N}\mathbf{\tilde{Y}}_{N}\mathbf{\tilde{Y}}%
_{N}^{T}+\alpha_{M}\mathbf{I}_{M}\right)  ^{-2}\boldsymbol{\upsilon}_{M}\\
\asymp\frac{1}{1-\gamma_{M}\tilde{\gamma}_{M}}\boldsymbol{\upsilon}_{M}%
^{T}\mathbf{R}_{M}^{1/2}\left(  \tilde{\delta}_{M}\mathbf{R}_{M}+\alpha
_{M}\mathbf{I}_{M}\right)  ^{-2}\mathbf{R}_{M}^{1/2}\boldsymbol{\upsilon}%
_{M}\text{,}%
\end{multline*}
Using (\ref{propAux2a}) in Proposition \ref{propAux} with $\Theta
_{M}=\boldsymbol{\upsilon}_{M}\boldsymbol{\upsilon}_{M}^{T}$ and
$z=-\alpha_{M}$, we have that (notice that $\left.  x_{M}-zx_{M}^{\prime
}\right\vert _{z=-\alpha_{M}}=\tilde{\delta}_{M}+\alpha_{M}\tilde{\zeta}_{M}$)%
\[
\left.  \frac{1}{N}\boldsymbol{\upsilon}_{M}^{T}\mathbf{\tilde{Y}}%
_{N}\mathbf{\tilde{Y}}_{N}^{T}\left(  \frac{1}{N}\mathbf{\tilde{Y}}%
_{N}\mathbf{\tilde{Y}}_{N}^{T}-z\mathbf{I}_{M}\right)  ^{-2}%
\boldsymbol{\upsilon}_{M}\right\vert _{z=-\alpha_{M}}\asymp\left(
\tilde{\delta}_{M}+\alpha_{M}\tilde{\zeta}_{M}\right)  \boldsymbol{\upsilon
}_{M}^{T}\mathbf{R}_{M}\left(  \tilde{\delta}_{M}\mathbf{R}_{M}+\alpha
_{M}\mathbf{I}_{M}\right)  ^{-2}\boldsymbol{\upsilon}_{M}\text{,}%
\]
and the proof follows by (\ref{derDeltatilde}) in Lemma \ref{lemmaAuxRel}.

Let us now handle (\ref{theoGCEb}). We want to prove that, in effect,%
\begin{multline*}
\frac{1}{\frac{1}{N}\operatorname*{tr}\left[  \mathbf{T}_{N}\left(
\mathbf{I}_{N}+\hat{\delta}_{M}\mathbf{T}_{N}\right)  ^{-2}\right]  }\left(
\boldsymbol{\upsilon}_{M}^{T}\left(  \frac{1}{N}\mathbf{\tilde{Y}}%
_{N}\mathbf{\tilde{Y}}_{N}^{T}\right)  ^{2}\left(  \frac{1}{N}\mathbf{\tilde
{Y}}_{N}\mathbf{\tilde{Y}}_{N}^{T}+\alpha_{M}\mathbf{I}_{M}\right)
^{-2}\boldsymbol{\upsilon}_{M}-\hat{\delta}_{M}^{2}\boldsymbol{\upsilon}%
_{M}^{T}\mathbf{T}^{2}\left(  \mathbf{I}_{N}+\hat{\delta}_{M}\mathbf{T}%
\right)  ^{-2}\boldsymbol{\upsilon}_{M}\right) \\
\asymp\frac{\gamma_{M}}{1-\gamma_{M}\tilde{\gamma}_{M}}\boldsymbol{\upsilon
}_{N}^{T}\left(  \delta_{M}\mathbf{T}_{N}+\mathbf{I}_{N}\right)
^{-2}\boldsymbol{\upsilon}_{N}\text{.}%
\end{multline*}
First, observe that%
\begin{multline*}
\boldsymbol{\upsilon}_{N}^{T}\left(  \frac{1}{N}\mathbf{\tilde{Y}}_{N}%
^{T}\mathbf{\tilde{Y}}_{N}\right)  ^{2}\left(  \frac{1}{N}\mathbf{\tilde{Y}%
}_{N}^{T}\mathbf{\tilde{Y}}_{N}+\alpha_{M}\mathbf{I}_{N}\right)
^{-2}\boldsymbol{\upsilon}_{N}\\
=\boldsymbol{\upsilon}_{N}^{T}\frac{1}{N}\mathbf{\tilde{Y}}_{N}^{T}%
\mathbf{\tilde{Y}}_{N}\left(  \frac{1}{N}\mathbf{\tilde{Y}}_{N}^{T}%
\mathbf{\tilde{Y}}_{N}+\alpha_{M}\mathbf{I}_{N}\right)  ^{-1}%
\boldsymbol{\upsilon}_{N}-\alpha_{M}\frac{1}{N}\boldsymbol{\upsilon}_{N}%
^{T}\mathbf{\tilde{Y}}_{N}^{T}\mathbf{\tilde{Y}}_{N}\left(  \frac{1}%
{N}\mathbf{\tilde{Y}}_{N}^{T}\mathbf{\tilde{Y}}_{N}+\alpha_{M}\mathbf{I}%
_{N}\right)  ^{-2}\boldsymbol{\upsilon}_{N}\text{.}%
\end{multline*}
Moreover, asymptotic deterministic equivalents of the two terms on the RHS can
be found by (\ref{propAux1b}) in Proposition \ref{propAux} with $\Theta
_{N}=\boldsymbol{\upsilon}_{N}\boldsymbol{\upsilon}_{N}^{T}$ and
$z=-\alpha_{M}$ as%
\[
\boldsymbol{\upsilon}_{N}^{T}\frac{1}{N}\mathbf{\tilde{Y}}_{N}^{T}%
\mathbf{\tilde{Y}}_{N}\left(  \frac{1}{N}\mathbf{\tilde{Y}}_{N}^{T}%
\mathbf{\tilde{Y}}_{N}+\alpha_{M}\mathbf{I}_{N}\right)  ^{-1}%
\boldsymbol{\upsilon}_{N}\asymp\delta_{M}\boldsymbol{\upsilon}_{N}%
\mathbf{T}\left(  \mathbf{I}_{N}+\delta_{M}\mathbf{T}\right)  ^{-1}%
\boldsymbol{\upsilon}_{N}%
\]
and%
\[
\frac{1}{N}\boldsymbol{\upsilon}_{N}^{T}\mathbf{\tilde{Y}}_{N}^{T}%
\mathbf{\tilde{Y}}_{N}\left(  \frac{1}{N}\mathbf{\tilde{Y}}_{N}^{T}%
\mathbf{\tilde{Y}}_{N}+\alpha_{M}\mathbf{I}_{N}\right)  ^{-2}%
\boldsymbol{\upsilon}_{N}\asymp\zeta_{M}\boldsymbol{\upsilon}_{N}%
^{T}\mathbf{T}\left(  \mathbf{I}_{N}+\delta_{M}\mathbf{T}\right)
^{-2}\boldsymbol{\upsilon}_{N}\text{.}%
\]
Then, by rearranging terms we can write%
\[
\boldsymbol{\upsilon}_{N}^{T}\left(  \mathbf{\tilde{Y}}_{N}^{T}\mathbf{\tilde
{Y}}_{N}\right)  ^{2}\left(  \mathbf{\tilde{Y}}_{N}^{T}\mathbf{\tilde{Y}}%
_{N}+\alpha_{M}\mathbf{I}_{N}\right)  ^{-2}\boldsymbol{\upsilon}_{N}%
\asymp\left(  \delta_{M}-\alpha_{M}\zeta_{M}\right)  \boldsymbol{\upsilon}%
_{N}^{T}\mathbf{T}\left(  \mathbf{I}_{N}+\delta_{M}\mathbf{T}\right)
^{-2}\boldsymbol{\upsilon}_{N}+\delta_{M}^{2}\boldsymbol{\upsilon}_{N}%
^{T}\mathbf{T}^{2}\left(  \mathbf{I}_{N}+\delta_{M}\mathbf{T}\right)
^{-2}\boldsymbol{\upsilon}_{N}\text{,}%
\]
and the result follows finally after straightforward algebraic manipulations
by (\ref{derDelta}) in Lemma \ref{lemmaAuxRel}.

\bibliographystyle{IEEEtran}
\bibliography{frubio2}

\end{document}